\documentclass[10pt,journal,compsoc]{IEEEtran}

\usepackage[nointegrals]{wasysym} 
\usepackage{multirow}
\usepackage{multicol}
\usepackage{footnote}
\usepackage{amsthm}
\usepackage{amsmath}
\usepackage{amssymb}
\usepackage{float}
\usepackage{bbm}
\usepackage{bm}
\usepackage{eucal}
\usepackage{mathrsfs}
\usepackage{dsfont}
\usepackage{mathtools}
\usepackage{booktabs}
\usepackage{colortbl}
\usepackage{cite}
\usepackage{extarrows}
\usepackage{graphicx}
\usepackage{setspace}
\usepackage{diagbox}
\usepackage{makecell}

\graphicspath{{figures/},{figures/photos/},{figures/experiments/}\:}
\DeclareGraphicsExtensions{.pdf,.jpg,.png}
\usepackage{pgf}
\usepackage{subfig}
\usepackage{indentfirst}
\allowdisplaybreaks[4]

\AtBeginDocument{\captionsetup[figure]{labelsep=period}}
\AtBeginDocument{\captionsetup[table]{aboveskip=0pt,labelsep=period}}
\renewcommand{\figurename}{Fig.}

\usepackage{color}
\usepackage{xcolor}
\DeclareCaptionFont{red}{\color{red}}

\usepackage{enumerate}
\usepackage[shortlabels]{enumitem}

\newtheorem{theorem}{\textbf{Theorem}}

\usepackage[noend]{algpseudocode}
\usepackage{algorithmicx,algorithm}

\floatstyle{boxed}
\newfloat{fig}{bt}{fig}
\floatname{fig}{Fig.}
\AtBeginDocument{\captionsetup[fig]{labelsep=period}}

\def\hlinew#1{%
    \noalign{\ifnum0=`}\fi\hrule \@height #1 \futurelet
\reserved@a\@xhline}

\usepackage{xspace}
\def\nameScheme{ESAFL\xspace}
\def\nameHE{ESHE\xspace}
\def\nameCompOne{BatchCrypt\xspace}
\def\nameCompTwo{FATE\xspace}

\ifCLASSINFOpdf
\else
\fi
\begin{document}
%
\title{\nameScheme: Efficient Secure Additively Homomorphic Encryption for Cross-Silo Federated Learning}

\author{Jiahui~Wu, Weizhe~Zhang, Fucai~Luo
\IEEEcompsocitemizethanks{\IEEEcompsocthanksitem
This work is supported by The Key Program of the Joint Fund of the National Natural Science Foundation of China (Grant No. U22A2036)
and The Key-Area Research and Development Program of Guangdong Province (Grant No. 2020B010136001).
\IEEEcompsocthanksitem Jiahui Wu, Weizhe Zhang, and Fucai Luo are with the Department of New Networks, Peng Cheng Laboratory, Shenzhen 518000, China
(e-mail: wjh01@pcl.ac.cn; weizhe.zhang@pcl.ac.cn; lfucai@126.com);
Corresponding author: Weizhe Zhang.
%

}
}

\IEEEtitleabstractindextext{%
\begin{abstract}
  Cross-silo federated learning (FL) enables multiple clients to collaboratively train a machine learning model without sharing training data, but privacy in FL remains a major challenge.
  Techniques using homomorphic encryption (HE) have been designed to solve this but bring their own challenges.
  Many techniques using single-key HE (SKHE) require clients to fully trust each other to prevent privacy disclosure between clients. However, fully trusted clients are hard to ensure in practice.
  Other techniques using multi-key HE (MKHE) aim to protect privacy from untrusted clients but lead to the disclosure of training results in public channels by untrusted third parties, e.g., the public cloud server.
  Besides, MKHE has higher computation and communication complexity compared with SKHE.
  We present a new FL protocol \nameScheme that leverages a novel efficient and secure additively HE (\nameHE) based on the hard problem of ring learning with errors.
  \nameScheme can ensure the security of training data between untrusted clients and protect the training results against
  untrusted third parties. 
  In addition, theoretical analyses present that \nameScheme outperforms current techniques using MKHE in computation and communication, and intensive experiments show that \nameScheme achieves approximate $204\times$-$953\times$ and $11\times$-$14\times$ training speedup while reducing the communication burden by $77\times$-$109\times$ and $1.25\times$-$2\times$ compared with the state-of-the-art FL models using SKHE.
\end{abstract}

}

\maketitle

\IEEEdisplaynontitleabstractindextext

%
\IEEEpeerreviewmaketitle

\IEEEraisesectionheading{\section{Introduction}\label{sec:introduction}}
Massive amounts of training datasets are required to build high-quality machine learning (ML) models.
In some areas like health care or finance, managing large datasets is often difficult and costly, not to mention that these datasets contain a large number of sensitive data (e.g., financial records, biometric health records, location information, shopping record) and cannot be disclosed \cite{2019FedMD}.
Cross-silo federated learning (FL) \cite{mcmahan2021advances} rises to this challenge.
It enables a group of distributed clients (such as medical and health institutions) and a central server to collaboratively train a model while keeping the clients' sensitive data local.
Specifically, in an iteration of FL, each client submits its local gradients to the server for model update.
The server aggregates the local gradients of all clients into global gradients and sends the global gradients back to the clients.
The clients then update their local model.
After numerous iterations, each client finally obtains a trained model.

To ensure that no client's update is revealed from its submitted gradients during data transfer and aggregation, many privacy-preserving techniques have been designed \cite{su2021secure,bonawitz2017practical,ben2022scionfl, mohamad2023sok,zheng2022aggregation,song2022eppda,shokri2015privacy,stevens2022efficient, DeepLearningHE, Zhang2020BatchCrypt,zhang2021towards,tian2021secure,chen2019efficient,ma2022privacy, Aloufi2020Blindfolded}. %
Among them, homomorphic encryption (HE) is particularly attractive to cross-silo FL, as it provides stronger privacy guarantees with no loss of learning accuracy.
In addition, compared with other privacy-preserving techniques, HE can greatly reduce the number of communication rounds in FL training, and it can be easily plugged into the existing FL frameworks to augment the popular parameter server architecture.
Traditional HE algorithms with single keys (SKHE) like Paillier and CKKS have been adopted in many FL applications \cite{DeepLearningHE, Zhang2020BatchCrypt,zhang2021towards}.
In these applications, all clients have the same encryption and decryption keys.
The clients encrypt their local gradients using the same encryption key and submit the encrypted ones to the cloud server. 
As all clients' gradients are encrypted, no information can be learned by the cloud server or external adversaries.
However, since the clients have the same decryption key, the ciphertexts of any client can be decrypted by a semi-honest client.
Therefore, FL applications using SKHE can be secure when all clients fully trust each other, but this is hard to ensure in practice.
To protect privacy against semi-honest clients, multi-key HE algorithms (MKHE) \cite{tian2021secure,chen2019efficient,ma2022privacy, Aloufi2020Blindfolded} are proposed and applied to FL, where clients encrypt their local gradients with different encryption keys.

Although federated learning based on MKHE is promising, two \textit{issues} still need to be well addressed.
\textbf{MKHE brings a new privacy threat to FL, namely that it cannot guarantee the security of the training results in the public channels against passive eavesdroppers \cite{lee2019security}.}
In MKHE, a distributed decryption protocol is required to decrypt a multi-key aggregation ciphertext.
The distributed decryption protocol consists of two steps:
\textit{Partial decryption}: Each client decrypts an aggregation ciphertext partially with its secret key and then outputs a partial decryption result;
\textit{Final decryption:} Each client operates all the clients' partial decryption results and obtains the aggregation plaintext.
In the protocol, the final decryption step takes only partial decryption results without any key.
In other words, the partial decryption results are not ciphertexts anymore.
Thus, any passive eavesdropping adversary can obtain the aggregation plaintext for free by merely getting the partial decryption results from the public channel and then implementing the final decryption step.
A naive solution is to use secure pair-wise communication channels to transmit partial decryption results between clients,
but as the number of clients and FL training rounds increases, the cost of establishing and maintaining secure channels may become very high.
\textbf{Second, MKHE owns low computational and communication efficiency.}
Most of the existing HE algorithms perform complex cryptographic operations such as modular multiplication and exponentiation which are extremely expensive to compute.
Usually, their encryption/decryption takes over $80\%$ of the training time, and encryption greatly expands the size of ciphertexts, which increases the amount of transmission data by more than $150\times$ compared with pure plaintext training \cite{Zhang2020BatchCrypt}.
Moreover, in MKHE, since the local gradient updates are encrypted with different keys, in order to generate correct aggregation gradient updates, the existing MKHE algorithms require additional interaction among clients in each model iteration. This brings extra significant computational and communication overhead to clients.

\textbf{Our Approach and Contributions.}
In this paper, we propose \nameScheme, an efficient and secure cross-silo federated learning solution to address the aforementioned issues.
Technically, \nameScheme leverages a novel efficient and secure additively HE algorithm (\nameHE) to protect the privacy of local gradients and training results against honest-but-curious entities. 
In summary, this paper contains the following \textit{contributions}.

\begin{itemize}[leftmargin=0.3cm]
 \item Based on the hard problem of ring learning with errors (RLWE) and using a secure pseudorandom generator, we propose a novel additively HE algorithm, \nameHE, which can provide privacy protection for local gradients and training results without the need for secure channel transmission of ciphertexts.
 We prove that \nameHE is IND-CPA secure and can resist collusion attacks of $\kappa<N-1$ untrusted clients.
 \nameHE only has one-step secure decryption, and this decryption step does not require interactions and secure channels between clients, so it has lower communication and computational complexity than existing MKHE algorithms.

\item We design a polynomial packing method to pack multiple local gradients into a single plaintext and utilize fast Fourier transform (FFT) for fast polynomial multiplication in plaintext encryption. These further improve the computational efficiency and reduce the communication traffic caused by the encryption algorithm of \nameHE.

\item We provide a complexity analysis and comparison of \nameScheme and the state-of-the-art privacy-preserving FL models using MKHE or SKHE.
The analysis results present that \nameScheme outperforms current FL models using MKHE and it has the same computation and communication complexity as FL models using SKHE.
\item We conduct intensive experiments on a fully-connected neural network (FCN) with FMNIST dataset, an ALexNet network with CIFAR10 dataset, and an LSTM network with Shakespeare dataset.
The experimental results show that \nameScheme achieves approximate $204\times$-$953\times$ and $11\times$-$14\times$ training speedup while reducing the communication burden by $77\times$-$109\times$ and $1.25\times$-$2\times$ compared with the state-of-the-art SKHE-based FL models.

\end{itemize}

\textbf{Paper Organization.} The remainder of this paper is organized as follows.
In Section~\ref{sec:RelatedWork}, we describe some related works about privacy-preserving federated learning.
In Section~\ref{sec:problemFormulation}, we give a formal description of our system model, threat model, and design goals.
Section~\ref{sec:Preliminaries} introduces the preliminaries of deep learning, federated learning, and related cryptographic knowledge.
We propose an efficient and secure homomorphic encryption algorithm (\nameHE) and design a polynomial packing method for \nameHE in Section~\ref{sec:MPHE}, then we construct \nameScheme model using \nameHE in Section~\ref{sec:wholeModel}.
Section~\ref{sec:Analysis} analyzes the security and the complexity of \nameScheme.
Section~\ref{sec:experiments} presents a performance evaluation of \nameScheme,
followed by a conclusion and future works in Section~\ref{sec:conclusion}.

\section{Related Work}\label{sec:RelatedWork}
%
%



%

Generally speaking, all the existing techniques to protect the privacy of clients in federated learning can be classified into four categories, as described below. 
%


\textbf{Secure aggregation (SA)} \cite{bonawitz2017practical} aims to ensure that the server learns no individual submissions from any clients but the aggregation results only.
While SA has been successfully deployed in cross-device FL \cite{ben2022scionfl, mohamad2023sok} in which the server hosts the training and obtains the trained model, it is unsuitable for 
cross-silo FL which requires that no one other than the clients knows the aggregation results. 
%
%
Moreover, security techniques in SA focus only on protecting the privacy of clients in a single training round \cite{so2021securing}. 
To ensure multi-round privacy security, in many SA-based models \cite{ben2022scionfl, mohamad2023sok,zheng2022aggregation,song2022eppda},  
clients must synchronize secret keys or zero-sum masks in each training iteration, which imposes a huge communication burden and is impractical for cross-silo FL.
%
%
%


\textbf{Differential privacy (DP)} usually facilitates privacy-preserving FL by adding tailored noises to clients' gradients. 
%
DP-based model in \cite{shokri2015privacy} proposes a distributed selective SGD with DP to make a trade-off between data privacy and model accuracy.
%
%
Although the model can be efficiently implemented, it is proven to be insecure under honest-but-curious server \cite{DeepLearningHE}.
%
DP-based model in \cite{stevens2022efficient} aims to protect data privacy against the curious server by combining DP with the hardness of learning with errors.
However, the model cannot ensure the security of the training results under the curious server.
Moreover, due to DP's noise injection, neither \cite{shokri2015privacy} nor \cite{stevens2022efficient} can achieve the best accuracy.
While these privacy violations and model accuracy degradations may be tolerable for mobile clients in cross-device FL, they raise serious concerns for institutions in cross-silo FL.


%

\textbf{Secure multiparty computation (SMC)} enables multiple parties to securely compute an agreed-upon function with their individual data. 
In this way, all parties know nothing except their own inputs and outputs.
SMC utilizes specifically designed computation and synchronization protocols among multiple parties. It has strong privacy guarantees of zero knowledge but is hard to achieve efficiently in a geo-distributed scenario like cross-silo FL \cite{Yang2019Federated}. 
ML algorithms must be specially adapted to suit the SMC paradigm \cite{Berry2004Privacy,Mohassel2018ABY,DeepLearnSMC2,damgaard2019new,dalskov2021fantastic,koti2022pentagod}.
To achieve better performance, privacy requirements may be lowered in some designs \cite{Berry2004Privacy,Mohassel2018ABY,DeepLearnSMC2}.
The state-of-the-art SMC protocols \cite{damgaard2019new,dalskov2021fantastic,koti2022pentagod} attempt to achieve strong security with high efficiency, but are only suitable to ML algorithms on a few clients (no more than five clients).

%
\textbf{Homomorphic encryption (HE)} allows performing certain computations (e.g., addition) directly on encrypted data. 
%
Many researchers design SKHE-based \cite{DeepLearningHE, Zhang2020BatchCrypt, zhang2021towards} or MKHE-based \cite{Mukherjee2016Two,chen2019efficient, Aloufi2020Blindfolded} cross-silo FL models to protect the privacy of the clients: each client submits its encrypted local updates to the central server for aggregating, and then downloads the aggregation result from the server for its local model update.
HE suits cross-silo FL for three reasons.
First, it can protect the privacy of the clients from untrusted third parties including the cloud server and external adversaries 
as the clients' submissions are encrypted. 
%
Second, it brings no loss in model accuracy, as no noise is injected into the clients' submissions during the whole training process.
Third, it requires no modifications of ML algorithms, as HE can be implemented as a plug-in module in the cross-silo FL framework.
%
Although HE is promising to FL training, the existing HE algorithms have shortcomings in protecting data privacy against internal clients or ensuring training result security from untrusted third parties.
%
Besides, both SKHE and MKHE introduce significant computation and communication burdens.

\textbf{In general,}
 each of these privacy protection technologies has its advantages and disadvantages. 
 MPC can provide powerful privacy protection but requires expert efforts to modify the existing ML algorithms to ensure the correctness of secure computation among distributed parties.
 DP can be used easily and efficiently, but provides weak privacy guarantees and introduces a loss in model accuracy. 
 SA is applicable to cross-device FL, but not to cross-silo FL, as it cannot guarantee the security of aggregation results and brings heavy synchronization costs. 
 HE can be simply plugged in cross-silo FL to provide strong privacy guarantees of clients' datasets against external parties
 without algorithm modifications or accuracy loss.
 However, without secure channels for transmitting ciphertexts, HE cannot protect data privacy from internal clients, and its variant algorithm MKHE aims to address this problem but brings new privacy threats to training results. 
 Besides, HE's computational and communication burden makes it less practical for production deployment.



 






\section{Problem Formulation}\label{sec:problemFormulation}
In this section, we introduce system model, threat model, and design goals of \nameScheme.

\subsection{System Model and Threat Model}\label{sssec:modelAssumption}

\textit{(1) System Model.}
The system of \nameScheme consists of
three entities: An administration server (AS), $N$ distributed clients, and a central cloud server.
The $N$ distributed clients are a small number of clients who collaborate to train their respective models with all the $N$ clients' datasets.
The cloud server is a coordinator to assist clients in collaboration.
The AS is an alternative entity \footnote{The clients can also generate parameters and keys by themselves using a secure distributed key generation method we introduce in Section 6: \nameScheme -- Initialization.}
that generates public parameters and keys for all clients in \nameScheme.
The general task flow of the three entities is as follows.
The $N$ clients establish their cooperation in training models with all clients' datasets and then send a request to AS and the cloud server.
Upon receiving the request from clients, AS initializes the neural network model, generates keys and parameters, and distributes them to $N$ clients.
After receiving an identical neural network model from AS, the $N$ clients separately train the network model with their respective datasets to generate local parameters.
The clients then submit their local parameters to the cloud server.
The cloud server aggregates all the received local parameters to a global parameter.
The cloud then returns the global parameter to the clients for local model updates.
After numerous client-cloud iterations, the network model can be trained finally.


\textit{(2) Threat Model.}
AS is considered to be a trusted entity (Alternatively, without the trusted AS, the clients can generate their keys using the existing MPC techniques \cite{ben2008fairplaymp}. For convenience, unless otherwise specified, we use AS in the subsequent parts of this paper).
The clients collaborate to train a model which is profitable to the clients themselves, so the clients are considered to be honest-but-curious.
That is, the clients perform the training processes honestly but are curious about the other clients' privacy.
The public cloud server that sells computing services is also honest-but-curious \cite{malekzadeh2021honest,liu2022privacy,wang2022privacy}.
Therefore, the possible threats come from the clients, the public cloud server, and outside adversaries who eavesdrop on the communication channels between the clients and the cloud server.
Each client may infer the other clients' privacy from the transferred gradients by using their shared knowledge.
The cloud server may recover the sensitive information of the clients and steal the training results/trained model from the data it can access.
The adversaries can eavesdrop on messages sent by the clients to the cloud from communication channels.
Thus, the adversaries can also recover the sensitive information of the clients and steal the model. %
Since the threat behaviors of the adversaries are the same as that of the cloud, in this paper, we only consider the threat from the public cloud and the clients.

In \nameScheme, we do not assume secure channels between clients and the cloud server.
All clients and the cloud transmit their messages on public channels.
The trusted AS is employed to generate and distribute secret keys and public parameters only in the initialization phase of the federated model.
Due to the fact that a federated model trained to converge generally requires thousands of iterations, the amortized cost (i.e., the average cost taken per training iteration) of AS is very small.
Therefore, the amortized cost of \nameScheme with the trusted AS assumption is less than that of FL models with the secure channels assumption 
that in each training iteration, the clients and the cloud communicate through secure pair-wise communication channels.

\subsection{Design Goals}

\nameScheme aims to achieve the following design goals.

\begin{itemize}[leftmargin=0.3cm]

  \item \textbf{Accuracy.} The prediction accuracy of ciphertext learning should not be reduced compared with plaintext learning.

  \item \textbf{Security.} The clients' private dataset should keep hidden from the other clients and the public cloud server.
    The aggregation results should be secret to the cloud server.

  \item \textbf{Efficiency.} The computation and communication burden on the clients and the cloud should be low.
\end{itemize}

\section{Preliminaries}\label{sec:Preliminaries}
In this section, we introduce the background knowledge of deep learning, cross-silo federated learning, homomorphic encryption, ring learning with errors problem, and CKKS's encoding method.
Table~\ref{tab:notations} summarizes mathematical notations used in this paper.

\subsection{Deep Learning}\label{ssec:Prelim_DL}

Deep learning is a machine learning method that trains deep neural network through iteratively performing feedforward and backpropagation processes, and the final trained neural network can be used to predict the output result of a given input.
The feedforward process can be expressed as $\hat{\boldsymbol{y}}=f(\boldsymbol{w},\boldsymbol{x})$, where $\boldsymbol{x}$ is an input vector, $\boldsymbol{w}$ is a weight parameter vector of the neural network, and $\hat{\boldsymbol{y}}$ is the predicted value of the input vector.
Given a loss function $l$, the training task is to minimize the loss
$$\mathcal{L}(\boldsymbol{w};D)=\dfrac{1}{|D|} \sum_{(\boldsymbol{x}_i,\boldsymbol{y}_i) \in D}
l(\boldsymbol{y}_i, f(\boldsymbol{x}_i,\boldsymbol{w})),$$ 
where $D=\{(\boldsymbol{x}_i,\boldsymbol{y}_i); i\in I\}$ represents the training dataset and $|D|$ is the size of the dataset.
In the backpropagation phase, the stochastic gradient descent (SGD) method is used to minimize the loss function and then update the weight parameters.
Specifically, in $t$th iteration, SGD finds the gradient
\begin{equation}\label{eq:gradientGen}
  \boldsymbol{g}^t= \nabla_{\boldsymbol{w}} \mathcal{L}\left(\boldsymbol{w}^t;D^t\right),
\end{equation}
where $D^t\subseteq D$ is a mini-batch of training dataset randomly selected from $D$, and $\nabla_{\boldsymbol{w}} \mathcal{L}\left(\boldsymbol{w}^t;D^t\right)$ represents the derivative of the loss with respect to the weight parameters.
Then, SGD updates the weight parameters as
\begin{equation}\label{eq:weightUpdate}
  \boldsymbol{w}^{t+1}\leftarrow \boldsymbol{w}^t-\eta\cdot \boldsymbol{g}^t,
\end{equation}
where the hyperparameter $\eta$ is the learning rate that determines the weight parameters' descent speed.
The weight parameters reach the optimum, step by step, according to the gradient.
Finally, the trained neural network is obtained.

\subsection{Cross-silo Federated Learning}

FL is a machine learning framework which enables a group of distributed clients to collaboratively train a
shared model while keeping the clients' private data locally. 
A coordination server is employed to orchestrate the training process.
A typical FL training process is presented as follows. 
First, a coordination server and $N$ clients initialize FL by deciding a training task, a shared global model $\boldsymbol{w}^0$, and the hyperparameters.
Then, the server and the clients collaborate to train a machine learning model by repeating the following steps until the final trained model is reached.
\begin{enumerate}
  \item \textbf{Client local training:} For each client $i$, it trains the global model $\boldsymbol{w}^t$ (where $t$ is the current iteration index) with its private dataset in local to obtain a local gradient $$\boldsymbol{g}_i^t = \nabla_{\boldsymbol{w}} \mathcal{L}\left(\boldsymbol{w}^t; D_i^t\right).$$
      Then, it submits the local gradient $\boldsymbol{g}_i^t$ to the server.

  \item \textbf{Server aggregation:} After receiving all clients' local gradients, the server aggregates them with an aggregation rule $f_{\rm{add}}$ to form a global gradient
  $$\boldsymbol{G}^t=f_{\rm{add}}(\boldsymbol{g}_1^t,\cdots,\boldsymbol{g}_N^t).$$
  In this paper, we utilize \texttt{FedAvg} \cite{mcmahan2017communication} as our federated aggregation method as it is the most commonly used, then the aggregation rule can be generalized as $$f_{\rm{add}}(\boldsymbol{g}_1^t,\cdots,\boldsymbol{g}_N^t)=\sum_{i=1}^N \alpha_i\boldsymbol{g}_i^t,$$
  where $\alpha_i=|D_i|$.
  %
      Then the server returns the global gradient $\boldsymbol{G}^t$ to the clients.

  \item \textbf{Client model update:} Each client updates its local model as $$\boldsymbol{w}_i^{t+1}=\boldsymbol{w}_i^{t}-\eta\cdot\frac{\boldsymbol{G}^t}{N}, \textrm{ for } i\in [N].$$
\end{enumerate}

According to different application scenarios, FL can be roughly divided into cross-device FL and cross-silo FL \cite{mcmahan2021advances}.
In cross-device FL, the clients are a very large number of mobile/edge devices with limited computing power and unreliable communications.
In cross-silo FL, the clients are a small number of companies/organizations with rich computing resources and reliable communication.

\subsection{Additively Homomorphic Encryption}\label{ssec:Prelim_AHE}

Additively homomorphic encryption (AHE) enables the arithmetic operation of addition to be performed on a number of ciphertexts which is equivalent to performing the same operations on the corresponding plaintexts.
Therefore, the privacy-sensitive information of the plaintexts can be preserved using AHE to encrypt and implement the ciphertext aggregation computations.
In general, AHE algorithm consists of four algorithms: (\texttt{GenKey}, \texttt{Enc}, \texttt{Dec}, \texttt{Eval}), which are a key generation algorithm, an encryption algorithm, a decryption algorithm, and an evaluation algorithm, respectively.
AHE algorithm is homomorphic, if $\texttt{Dec}(dk,\texttt{Eval}(\texttt{Enc}(ek,m_1),\cdots,\texttt{Enc}(ek,m_N)))=\texttt{Dec}(dk,\texttt{Enc}(ek,\sum_{i=1}^N m_i))$ holds, where $\{m_i\}_{i\in[N]}$ is a set of plaintexts, and $ek$ and $dk$ are generated by \texttt{KeyGen} algorithm, representing the encryption key and the decryption key, respectively.
Its potential meaning is that the aggregation result of the plaintexts can be obtained by evaluating their corresponding ciphertexts, rather than directly calculating the plaintexts, so as to realize secure computation of addition.

\subsection{Ring Learning with Errors}\label{ssec:Prelim_RLWE}

RLWE is a computational problem that serves as the foundation of quantum-resistant cryptographic algorithms and also provides the basis for homomorphic encryption, such as BFV \cite{fan2012somewhat} and CKKS \cite{cheon2017homomorphic}.
RLWE problem is built on the arithmetic of polynomials with coefficients from a finite field.
Let $n$ be a power of two. We denote by $\mathcal{R}=\mathbb{Z}[x]/(X^n+1)$ the ring of integers of the $(2n)$th cyclotomic field and $\mathcal{R}_q=\mathbb{Z}_q[x]/(X^n+1)$ the residue ring of $\mathcal{R}$ modulo an integer $q$.
An element of $\mathcal{R}$ (or $\mathcal{R}_q$) is a polynomial $a(x)=\sum_{i=0}^{n-1} a_i\cdot x^i$ and can be identified with a vector of coefficients $(a_0,\cdots,a_{n-1})$ in $\mathbb{Z}^n$ (or $\mathbb{Z}_q^n$).
Let $\chi_{s}$ and $\chi_{e}$ be distributions over $\mathcal{R}$ which output small polynomials, i.e., for $s\leftarrow \chi_{s},e\leftarrow\chi_{e}$ and small bounds $B_s$ and $B_e$, there are $\|s\|_{\infty}\leq B_s$ and $\|e\|_{\infty}\leq B_e$ with an overwhelming probability.
A typical RLWE assumption states that the RLWE samples $(b,a)\in \mathcal{R}_q^2$ are computationally indistinguishable from uniformly random elements of $U(\mathcal{R}_q^2)$, where $b=[-a\cdot s+e]_q, a\leftarrow U(\mathcal{R}_q), e\leftarrow \chi_e, s\leftarrow \chi_s$.

\subsection{CKKS's Message Encoding and Decoding}

The plaintext space of RLWE-based HE schemes is usually a cyclotomic polynomial ring $\mathbb{Z}_q[X]/(X^n+1)$ of a finite characteristic.
However, the messages in the federated learning that need to be encrypted are real numbers.
Therefore, we should encode real numbers as integral coefficient plaintexts in the ring $\mathbb{Z}_q[X]/(X^n+1)$ before encrypting.
CKKS scheme proposes an efficient encoding and decoding method for transforming between complex vectors and integral coefficient polynomials.
With this encoding method, operations on the encoded polynomials can be accelerated by using FFT.
CKKS's encoding method $\texttt{Ecd}(\boldsymbol{z})\rightarrow m(X)$ encodes a complex message vector $\boldsymbol{z}=(z_1,\cdots,z_{n/2})\in \mathbb{C}^{n/2}$ as a plaintext polynomial $m(X)\in \mathcal{R}_q:=\mathbb{Z}_q/(X^n+1)$ by computing $m(X)=\left\lfloor \Delta\cdot\phi^{-1}(\boldsymbol{z})\right\rceil\in \mathcal{R}_q$, where $\Delta>1$ is a scaling factor and $\phi$ is a ring isomorphism ${\phi:\mathbb{R}[X]/(X^{n}+1)\rightarrow \mathbb{C}^{n/2}}$.
Its corresponding decoding method $\texttt{Dcd}(m(X))\rightarrow \boldsymbol{z}$ decodes a plaintext polynomial $m(X)\in \mathcal{R}_q$ to a complex vector $\boldsymbol{z}\in\mathbb{C}^{n/2}$ by computing $\boldsymbol{z}=\Delta^{-1}\cdot\phi(m(X))\in\mathbb{C}^{n/2}$.
The scaling factor $\Delta>1$ is utilized to control the encoding/decoding error which is occurred by the rounding process.
Therefore, the approximate equation $\texttt{Dcd}(\texttt{Ecd}(\boldsymbol{z};\Delta);\Delta)\approx\boldsymbol{z}$ can be obtained by controlling $\Delta$.
The ring-isomorphic property of the mapping ${\phi:\mathbb{R}[X]/(X^{n}+1)\rightarrow \mathbb{C}^{n/2}}$ can ensure homomorphic computing of addition, i.e.,
there are $\texttt{Dcd}(\texttt{Ecd}(\boldsymbol{z}_1)+\texttt{Ecd}(\boldsymbol{z}_1);\Delta)\approx \boldsymbol{z}_1+\boldsymbol{z}_2$.
Therefore, when $\Delta$ is chosen appropriately, the approximate correctness of the additions in the encoded state can be guaranteed.

\begin{table}
 \footnotesize
 \renewcommand\arraystretch{1.23}
 \caption{The notations and their semantic meanings.}\label{tab:notations}
 \begin{tabular}{m{2cm}|m{6.2cm}}
  \hline
  \textbf{Notations}   & \textbf{Meanings}\\\hline \hline
  $\boldsymbol{u},\boldsymbol{v},\cdots$ & Vectors in bold letter\\
  \hline
  $N$ & The number of the clients \\
  \hline
  $D, |D|$ & The training dataset and its size\\
  \hline
  $(\boldsymbol{x},\boldsymbol{y})$ & A pair of training example (input data and its truth label) sampled from $D$\\
  \hline
  $f$ & Network function\\
  \hline
  $\hat{\boldsymbol{y}}$ & Predicted value of the neural network\\
  \hline
  $\hat{\boldsymbol{w}}^t$ & A neural network model of the $t$th iteration\\
  \hline
  $g^{ij}$ & The gradient of the $i$th output in the $j$th layer of the neural network \\
  \hline
  $\boldsymbol{g}$ & The (local) gradient parameter which is the flattened vector of all gradient matrixes in the neural network, i.e., $\boldsymbol{g}=[g^{11},\cdots,g^{ij}]\in \mathbb{R}^{ij}$\\
  \hline
  $\mathcal{L}(\boldsymbol{w};(\boldsymbol{x},\boldsymbol{y}))$ & Loss function on the training example $(\boldsymbol{x},\boldsymbol{y})$\\
  \hline
  $\eta$ & Learning rate\\
  \hline
  $i\in[n]$ & $i=1,\cdots,n$ \\
  \hline
  $i\in[m,n]$ & $i=m,\cdots,n$ \\
  \hline
  $[0,n]^k$ or $[0,n)^k$ & A base-$n$ or base-$(n-1)$ integer with $k$ bits\\
  \hline
  $[z]_q$ & The reduction of the integer $z$ modulo $q$\\
  \hline
  $(z)_2$ & The binary representation of the integer $z$\\
  \hline
  $\|\cdot\|$ & Infinity norm\\
  \hline
  $\lfloor \cdot \rceil$ & The coefficient-wise rounding function\\
  \hline
  $\lambda$ & Security parameter\\
  \hline
  $X(\lambda)$ & An integer $X$ satisfies the security level $\lambda$\\
  \hline
  $\mathbb{R},\mathbb{C},\mathbb{Z},\mathbb{Z}_q$ & The set of real numbers, the set of complex numbers, the set of integers, and the set of $\mathbb{Z}$ modulo an integer $q$, respectively\\
  \hline
 $\mathcal{R}=\mathbb{Z}[x]/(X^n+1)$ & The ring of integers of the $(2n)$th cyclotomic field with $n$ the power of two\\
 \hline
 $\mathcal{R}_q$ & The residue ring of $\mathcal{R}$ modulo an integer $q$\\
 \hline
  $B_s,B_e,B_r$ & Small bounds\\
  \hline
  $\chi_{s}, \chi_{e}$ and $\chi_{r}$ & Distributions over $\mathcal{R}$ which output polynomials $s,e,r$ meeting $\|s\|_{\infty}^{\texttt{can}}\leq B_s$, $\|e\|_{\infty}^{\texttt{can}}\leq B_e$, $\|r\|_{\infty}^{\texttt{can}}\leq B_r$\\
  \hline
  $U(\cdot)$ & Uniform distribution\\
  \hline
  $*\leftarrow \chi$ & Sampling from a set or a distribution $\chi$ \\
  \hline
  $s_i, i\in [N]$ & The encryption keys of the $N$ clients\\
  \hline
  $s$ & The decryption key\\
  \hline
  $\log$ & The binary logarithm, i.e., the logarithm to the base 2\\
  \hline
 \end{tabular}
\end{table}

\section{\nameHE Design}\label{sec:MPHE}
This section introduces the design of \nameHE which is used to protect privacy-sensitive information in cross-silo FL.

\subsection{Overview of \nameHE}\label{ssec:overview}


Here are four important considerations for \nameHE design to achieve its security and efficiency. 
\begin{enumerate}[leftmargin=0.3cm]
  \item \textbf{For privacy protection of the clients' local gradients under honest-but-curious entities.} 
  In cross-silo FL, curious entities (the clients and the cloud) may infer the privacy of (other) clients from the information transmitted during training. 
    Therefore, \nameHE should support secure aggregation to prevent this privacy leakage. 
    The design of \nameHE mainly lies in the following two folds.
    \begin{itemize}
        \item We design the encryption algorithm of \nameHE in the form of multiple keys, in which clients encrypt their plaintexts with different secret keys.
        Thus, any curious entity cannot infer the (other) clients' plaintexts from the transmission ciphertexts it eavesdrops on from public channels.
        Besides, the clients can correctly decrypt a ciphertext only when the ciphertext is an aggregation of all clients' local ciphertexts; Otherwise, an aggregation ciphertext of partial clients' local ciphertexts cannot be correctly decrypted.
        Thus, \nameHE can prevent privacy leakage from publicly shared information in FL and it is robust to collusion between $\kappa<N-1$ clients.

    \item The secure aggregation of \nameHE should ensure both security and efficiency over multiple training rounds in FL.
    The most common solution of secure aggregation to ensure multi-round privacy security is enabling all clients in FL to regenerate the zero-sum masks/secret keys of its privacy protection algorithm. 
    In each iteration, it requires the clients to coordinate to generate and then share the masks/keys in secure channels, which significantly increases the cost of mask/key generation and communication.
    Different from the common solution, we turn to generate a primary random element in the encryption algorithm to replace each round of secret key generation and enable the clients to generate the same primary random element locally, to provide multi-round privacy guarantees while requiring no extra communication burden.

    \end{itemize}

  \item \textbf{For security of the decryption result under the honest-but-curious public cloud server.}
  	Considering that the distributed decryption protocol is not secure in public channels, as the final decryption process does not need any key and anyone can obtain the decryption result for free, we thus design 
   that decryption can only be performed locally by clients who have the decryption key.
    %
    Thereby ensuring that the decryption result is not obtained directly or through inferring.

  \item \textbf{For communication efficiency.}
    The previous MKHE algorithms require two-step decryption, i.e., partially decrypted and fully decrypted.
    The partial decryption process consumes large communication traffic because it requires all clients to decrypt an aggregation ciphertext separately with their respective keys, and then interact with each other to transmit their partial decryption results which are still ciphertexts. 
    Therefore, we tend to design a one-step decryption protocol that removes the partial decryption step, to improve communication efficiency.
    With our one-step decryption protocol, the clients who have a decryption key can decrypt an aggregation ciphertext locally without interaction.

  \item \textbf{For computational efficiency.}
  After removing the partial decryption step, the computation cost of \nameHE mainly lies in encryption for the RLWE-based cryptography.
  To improve the efficiency of encryption, we employ CKKS' encoding method to encode clients' gradients into polynomials and design a polynomial packing technique to pack multiple polynomials into a single plaintext.
  With the use of the encoding and the packing, the number of plaintexts to be encrypted is reduced by times, so the computation cost of encryption and the communication traffic of ciphertext in the whole model is reduced.
  Further, FFT can be utilized for fast polynomial multiplication in the encryption algorithm of \nameHE.
\end{enumerate}

\subsection{\nameHE Design}

In this section, we first introduce the design of \nameHE which is used to protect the privacy of local gradients and training results in FL.
Then we present how multiple messages can be encoded and packed into a single plaintext to improve the efficiency of \nameHE.

\subsubsection{Encryption Algorithm of \nameHE}\label{sssec:MKHE}
We design \nameHE as a symmetric HE algorithm based on the hard problem of RLWE.
Specifically, \nameHE consists of a tuple of six algorithms denoted as $\Pi_{HE}=(\texttt{HE.Setup},\texttt{HE.KeyGen},\texttt{HE.Enc},\texttt{HE.Dec},\texttt{HE.Eval},$ $\texttt{HE.PRPG})$.
\begin{itemize}[leftmargin=0.3cm]
  \item[$\bullet$] \texttt{HE.Setup}$(1^{\lambda})\rightarrow \boldsymbol{PP}$: Setup algorithm \texttt{HE.Setup} takes as input a security parameter $\lambda$ and outputs public parameters $\boldsymbol{PP}=\{n,p,q,\chi_s,\chi_e\}$, where $n$ is the RLWE dimension, $q=q(\lambda)$ is a power-of-two integer, $p<q$ is an integer, and $\chi_s$ and $\chi_e$ are secret key distribution and error distribution, respectively.

  \item[$\bullet$]
    $\texttt{HE.KeyGen}(\boldsymbol{PP})\rightarrow (\{s_i\}_{i\in[N]},s,a^0,B)$: The key generation algorithm \texttt{HE.KeyGen} takes as input the public parameters $\boldsymbol{PP}$ generated in \texttt{HE.Setup}, and outputs $N$ secret/encryption keys $\{s_i\}_{i\in[N]}$, a decryption key $s$, a random public polynomial $a^0$, and a random secret parameter $B$.
     Then, it distributes $\{a^0,s_i,s\}$ to the $i$th client through a secure channel.
     Thus, the different clients own different encryption keys and the same decryption key.
     Note that $a^0$ and $s$ are the same for all the clients. The specific values of the parameters and the keys are as follows.
      \begin{itemize}
        \item The public parameter $a^0$ is uniformly sampled from $\mathcal{R}_q$.
        \item The $N$ encryption keys $\{s_i\}_{i\in[N]}$ have the same distribution. We sample the $N$ encryption keys from $\chi_s$.%
        \item The decryption key is the sum of the $N$ encryption keys, i.e., the decryption key is $s=\left[\sum_{i=1}^Ns_i\right]_q$.
        \item The random secret parameter $B$ is uniformly sampled from $[0,1]^k$, where $k$ is equal to or greater than the number of binary bits of the maximum iteration of FL training.
      \end{itemize}

  \item[$\bullet$]
    $\texttt{HE.Enc}(a^t,s_i,m_i^t)\rightarrow \boldsymbol{c}_i^t$: Encryption algorithm \texttt{HE.Enc} encrypts a plaintext polynomial $m_i^t$ using the encryption key $s_i$ and the random public polynomial $a^t\in\mathcal{R}_q$, where $t$ represents the current iteration number.
    The algorithm then outputs a ciphertext polynomial
    \begin{equation}\label{eq:Encryption}
        \boldsymbol{c}_i^t \leftarrow \texttt{HE.Enc}(a^t,s_i, m_i^t) = \left[a^t\cdot s_i+\tilde{p}\cdot e_i^t+m_i^t\right]_q,
    \end{equation}
    where $e_i\leftarrow \chi_e$ and $\tilde{p}$ is a polynomial of degree $n-1$ with all coefficients of $p$.

  \item[$\bullet$] $\texttt{HE.Eval}\left(\{\boldsymbol{c}_i^t\}_{i\in[N]}\right)\rightarrow \boldsymbol{C}^t$:
      Evaluation algorithm \texttt{HE.Eval} adds $N$ ciphertexts $\{\boldsymbol{c}_i^t\}_{i\in[N]}$, and outputs an aggregation ciphertext
      \begin{equation}\label{eq:aggregation2}
        \begin{split}
          \boldsymbol{C}^t & \leftarrow \texttt{HE.Eval}\left(\{\boldsymbol{c}_i^t\}_{i\in[N]}\right)
          =\sum_{i=1}^N \boldsymbol{c}_i^t\\
          & = \left[a^t\cdot \sum_{i=1}^N s_i + \tilde{p}\cdot \sum_{i=1}^N e_i^t + \sum_{i=1}^N m_i^t\right]_q\\
            & = \left[a^t\cdot s + \tilde{p}\cdot e^t + m^t_{\textrm{add}}\right]_q,
        \end{split}
      \end{equation}
      where $e^t=\sum_{i=1}^N e_i^t\in \mathcal{R}_q$ and $m^t_{\textrm{add}}=\sum_{i=1}^N m_i^t\in \mathcal{R}_q$.

  \item[$\bullet$]
    $\texttt{HE.Dec}(a^t, s,\boldsymbol{C}^t)\rightarrow M^t$: Decryption algorithm \texttt{HE.Dec} decrypts an aggregation ciphertext $\boldsymbol{C}^t$ of the $t$th iteration using the decryption key $s$ and outputs the aggregation plaintext
      \begin{equation}\label{eq:decryption}
        \begin{split}
          M^t &\leftarrow \texttt{HE.Dec}(a^t, s,\boldsymbol{C}^t) = \left[\left[\boldsymbol{C}^t-a^t\cdot s\right]_q\right]_p\\
            & = \left[\left[ \left[a^t\cdot s + \tilde{p}\cdot e + m_{\textrm{add}}\right]_q - [a^t\cdot s]_q \right]_q\right]_p\\
            & = m_{\textrm{add}},
        \end{split}
      \end{equation}

  \item[$\bullet$]
    $\texttt{HE.PRPG}(t+1,B)\rightarrow a^{t+1}$: The pseudorandom polynomial generation algorithm \texttt{HE.PRPG} takes as inputs the iteration number $t+1$  and the secret parameter $B$, then it outputs a pseudorandom polynomial $a^{t+1}\in\mathcal{R}_q$ as
    \begin{equation}
     \begin{split}
        a^{t+1} &\leftarrow \texttt{HE.PRPG}(t+1,B) \\
        &=\texttt{Vec2Poly}\left(\texttt{PRG}\left((x_1 \oplus b_1)||\cdots|| (x_k\oplus b_k)\right)\right), 
     \end{split}
    \end{equation}
    where ``$||$" represents the concatenation of binary digits and $(x_1||\cdots|| x_k)_2=(x_1\cdots x_k)_2=t+1$ and $(b_1||\cdots|| b_k)_2=(b_1\cdots b_k)_2=B$ hold; $\texttt{PRG}$ is a secure pseudorandom generator defined in \cite{bonawitz2017practical} with input space $[0,1]^{k}$ and output space $[0,q)^{n}$; $\left((x_1 \oplus b_1)||\cdots|| (x_k\oplus b_k)\right)$ is the seed of $\texttt{PRG}$;
    \texttt{Vec2Poly} transforms the $n$-dimension output vector of \texttt{PRG} to a polynomial in $\mathcal{R}_q$.
    Using the pseudorandom polynomial generation algorithm \texttt{HE.PRPG}, in each training iteration $t$, the clients can generate the same primary random element $a^t$ locally for the encryption algorithm \texttt{HE.Enc} without interaction,
    thus providing multi-round privacy guarantees while requiring no extra communication burden.
\end{itemize}

\textbf{Security.} \nameHE is IND-CPA secure and it can resist collusion attacks between $\kappa<N-1$ clients.
Besides, the decryption algorithm of \nameHE requires only one-step decryption and the decryption results can only be obtained by the one who has the decryption key, so the security of decryption results can be ensured. 
We next prove the security of \nameHE.

(1) \nameHE is IND-CPA secure 
under the RLWE assumption of parameter $(n,q, \chi_s,\chi_e)$ and the security/randomness of \texttt{PRG}. We prove it by showing that the distribution
\begin{equation*}
    \begin{split}
        \{(B,a^0,a^t, \boldsymbol{c}_i^t):&\boldsymbol{PP}=\{n,p,q,\chi_s,\chi_e\}\leftarrow \texttt{HE.Setup}(1^{\lambda}),\\
         &(\{s_i\}_{i\in[N]},s,a^0,B)\leftarrow\texttt{HE.KeyGen}(\boldsymbol{PP}),\\
         &a^t\leftarrow \texttt{HE.PRPG}(t,B),\\
         &\boldsymbol{c}_i^t \leftarrow \texttt{HE.Enc}(a^t,s_i, m_i^t)\} \\
    \end{split}
\end{equation*}
is computationally indistinguishable from the uniform distribution over $[0,1]^k\times\mathcal{R}_q\times\mathcal{R}_q\times\mathcal{R}_q$ for an arbitrary $m_i^t\in\mathcal{R}_p$.
The proof of this result is presented in Section~\ref{ssec:secAna}: Security Analysis -- the proof of Theorem~\ref{the:IND-CPA}.

(2) \nameHE can provide multi-round privacy guarantees under the RLWE assumption of parameter $(n,q, \chi_s,\chi_e)$ and the randomness of the primary random element $a^t, (t=1,2,\cdots)$.
We prove it by presenting that the distribution
\begin{equation*}
    \begin{split}
       &\{(a^t,a^{\tau},\tilde{\boldsymbol{C}}^t,\tilde{M}^t):\\
       & a^t\leftarrow\texttt{HE.PRPG}(t,B),a^{\tau}\leftarrow\texttt{HE.PRPG}(\tau,B),\\
       &\tilde{\boldsymbol{C}}^t\leftarrow[\boldsymbol{C}^t-\boldsymbol{c}_i^t+\boldsymbol{c}^{\tau}_i]_q,\\
       & \tilde{M}^t \leftarrow \texttt{HE.Dec}(a^t, s,\tilde{\boldsymbol{C}}^t)\} 
    \end{split}
\end{equation*}
is computationally indistinguishable from the uniform distribution over $\mathcal{R}_q^3\times\mathcal{R}_p$ for arbitrary spanning-rounds aggregation ciphertext $\tilde{\boldsymbol{C}}$ (which is an aggregation of $N$ local ciphertexts spanning two rounds) and its $N$ original plaintexts.
The above proof can be extended to arbitrary spanning-rounds aggregation ciphertext that is an aggregation of $N$ local ciphertexts spanning multiple rounds.
This presents that the decrypted plaintext of arbitrary spanning-rounds aggregation ciphertext is a random polynomial in $\mathcal{R}_p$.
The proof of this result is presented in Section~\ref{ssec:secAna}: Security Analysis -- the proof of Theorem~\ref{the:SecMultiRound}.

(3) \nameHE is secure against collusion attacks of $\kappa<N-1$ clients. 
We prove it by presenting that in the worst-case scenario, i.e., collusion occurs between $N-2$ clients (assume the colluded $N-2$ clients are the client 1 to $N-2$), 
the distribution of an arbitrary honest client's secrets
\begin{equation*}
\footnotesize
    \begin{split}
       &\{(s_i,e_i,m_i), i\in\{N-1,N\}: \\
       &s_j\leftarrow\chi_s (j\in[N]),s_{N-1}+s_N = s - \sum_{k=1}^{N-2}s_k,\\
       &a\leftarrow U(\mathcal{R}_p),e_j\leftarrow\chi_e,
        \boldsymbol{c}_j\leftarrow \texttt{HE.Enc}(a,s_j, m_j), \\
        &\boldsymbol{C}\leftarrow\texttt{HE.Eval}\left(\{\boldsymbol{c}_j\}_{j\in[N]}\right),
        M\leftarrow\texttt{HE.Dec}(a,s,\boldsymbol{C}), \\
        &e_{N-1}+e_N =[\boldsymbol{c}_{N-1}+\boldsymbol{c}_N - a(s_{N-1}+s_N) - (M - \sum_{k=1}^{N-2}m_k)]_q|p,\\
       &m_{N-1}+m_N = M - \sum_{k=1}^{N-2}m_k\}
    \end{split}
\end{equation*}
is computationally indistinguishable from the distribution $\left(\chi_s,\chi_e,U(\mathcal{R}_p)\right)$. 
The proof of this result is presented in Section~\ref{ssec:secAna}: Security Analysis -- the proof of Theorem~\ref{the:SecCollusion}.

\textbf{Efficiency.}
The efficiency improvement of \nameHE mainly lies in three folds.
(1) \nameHE provides privacy protection but does not require secure channels to transmit ciphertexts and decryption results while existing SKHEs and MKHEs require secure channels; otherwise, their original plaintexts or decryption results will be leaked from their transmission messages.
(2) The decryption algorithm of \nameHE has only one-step decryption and requires no client interactions.
Thus, both computation and communication efficiencies of \nameHE are higher than that of existing MKHEs whose decryption algorithm contains two-step decryption and requires interactions between clients.
Besides, the newly added pseudorandom polynomial generation algorithm \texttt{HE.PRPG} does not bring extra communication burden for \nameHE, because each client generates the same pseudorandom polynomial $a^t$ locally without interaction.
The computational cost of generating $a^t$ can be ignored compared to the computational cost of encryption and decryption.
(3) Messages encoding and polynomials packing method is designed to further improve the efficiency of \nameHE in Section~\ref{ssec:encoding}.
We give theoretical complexity analyses and comparisons of the FL model using \nameHE and the FL models using the prior SKHEs and MKHEs in Section~\ref{ssec:compAna}: Complexity Analysis.

\subsubsection{Messages Encoding and Polynomials Packing}\label{ssec:encoding}
To enable \nameHE to support homomorphic operations over real numbers and reduce the number of polynomials to be encrypted, 
we encode real numbers as integral coefficient polynomials using CKKS's encoding method and design a polynomial packing method to pack multiple encoded polynomials into a single plaintext polynomial for encrypting.
Messages encoding and polynomials packing mechanism for \nameHE consists of a tuple of four algorithms denoted as $\Pi_{EP}=\{\texttt{EP.Setup},\texttt{EP.EcdPack},\texttt{EP.DcdUnpk},\texttt{EP.Eval}\}$.
\begin{itemize}[leftmargin=0.3cm]
  \item[$\bullet$] $\texttt{EP.Setup}(\boldsymbol{PP},\log q_0,N) \rightarrow (T,pad)$:
      Setup algorithm \texttt{EP.Setup} takes as inputs the public parameters $\boldsymbol{PP}=\{1^{\lambda},p,q,\chi_s,\chi_e\}$ of \nameHE, the precision of messages $\log q_0$, and the number of clients $N$. 
      It outputs a zero-padding number $pad\geq\left\lceil\log  N\right\rceil+1$ and a slot numbers $T\leq\left\lfloor \frac{\log p}{\log q_0+ pad} \right\rfloor$
      as the batch size of polynomials for packing.

  \item[$\bullet$] $\texttt{EP.EcdPack}(q_0,p,T,pad,\boldsymbol{g})\rightarrow m(X)$: Messages encoding and packing algorithm \texttt{EP.EcdPack} takes as inputs $q_0,p,T,pad$ and a message vector $\boldsymbol{g}\in\mathbb{R}^L$, where $L$ is the length of the vector.
      It outputs a packed plaintext polynomial set $\{m(X)\in \mathcal{R}_{p}\}$.
      An example of messages encoding and then packing is presented in Fig.~\ref{fig:packing}.
      Specifically, \texttt{EP.EcdPack} algorithm performs the following two algorithms.
      %
      \begin{enumerate}[i), leftmargin=0.4cm]
        \item Messages encoding: This algorithm encodes each $n$ messages of the vector $\boldsymbol{g}$ as an $(n-1)$-degree polynomial in ring $\mathcal{R}_{q_0}$.
            We pad $0$ at the end of $\boldsymbol{g}$ if there are more than 1 and less than $n$ messages left without being chosen. 
            We denote the $n$ messages chosen in the $i$th time as $\{g_{i,j}\}_{j\in[n]}$.
            The $n$ messages are first converted to a complex vector $\boldsymbol{z}_i\in\mathbb{C}^{n/2}$, then the complex vector is encoded as an $(n-1)$-degree integral coefficient polynomial
             \begin{equation}
               \begin{split}
                 z_i(X)=
                    & \left[\left\lfloor
                        p\cdot\phi^{-1}(\boldsymbol{z}_i)
                      \right\rceil\right]_{q_0} \\
                   =& \sum_{j=0}^{n-1}\alpha_{i,j}X^j,
                   \textrm{ for } i\in[T]
               \end{split}
             \end{equation}
             where $z_i(X)\in\mathcal{R}_{q_0}$ and $\phi: \mathbb{R}[X]/(X^n+1)\rightarrow \mathbb{C}^{n/2}$ is a ring isomorphism constructed in CKKS scheme.

        \item Polynomials packing: This algorithm packs each $T-1$ encoded message polynomials $z_i(X)$ ($i\in [T-1]$) as well as a zero polynomial $\vartheta$ of degree $(n-1)$ into a single polynomial
          \begin{equation}
              m(X) = \sum_{j=0}^{n-1}\beta_{j}X^j,
          \end{equation}
          where $m(X)\in\mathcal{R}_p$, and $\beta_{j}$ is the $j$th coefficient of the packed polynomial $m(X)$. The binary representation of $\beta_{j}$ is
           \begin{align}\label{eq:packing}
            \textcolor[RGB]{255,0,0}{\beta_{j}} =\ &
            \Big(
              \overbrace{
                \underbrace{0_{pad}\textcolor[RGB]{132,12,24}{\alpha_{1,j}}}_{\left(pad+\log q_0\right) \textrm{ bits}}^{\textrm{the } 1 \textrm{st slot}}
                \cdots
                \underbrace{0_{pad}\textcolor[RGB]{132,12,24}{\alpha_{T-1,j}}}_{\left(pad+\log q_0\right) \textrm{ bits}}^{\textrm{the } (T-1) \textrm{th slot}}
              \quad
              \underbrace{{0_{pad}\textcolor[RGB]{33,89,104}{\vartheta_j}}}_{\left(pad+\log q_0\right) \textrm{ bits}}^{\textrm{the } T \textrm{th slot}}
              }^{T \left(pad+\log q_0\right) \textrm{ bits}}
            \Big)_2,\notag\\
                   & \textrm{ for } j\in[0,n-1],
           \end{align}
      \end{enumerate}
      where $\vartheta_j=0$ is the $j$th coefficient of the zero polynomial $\vartheta$.

\begin{figure}[h]
  \centering
  \includegraphics[width=3.5in]{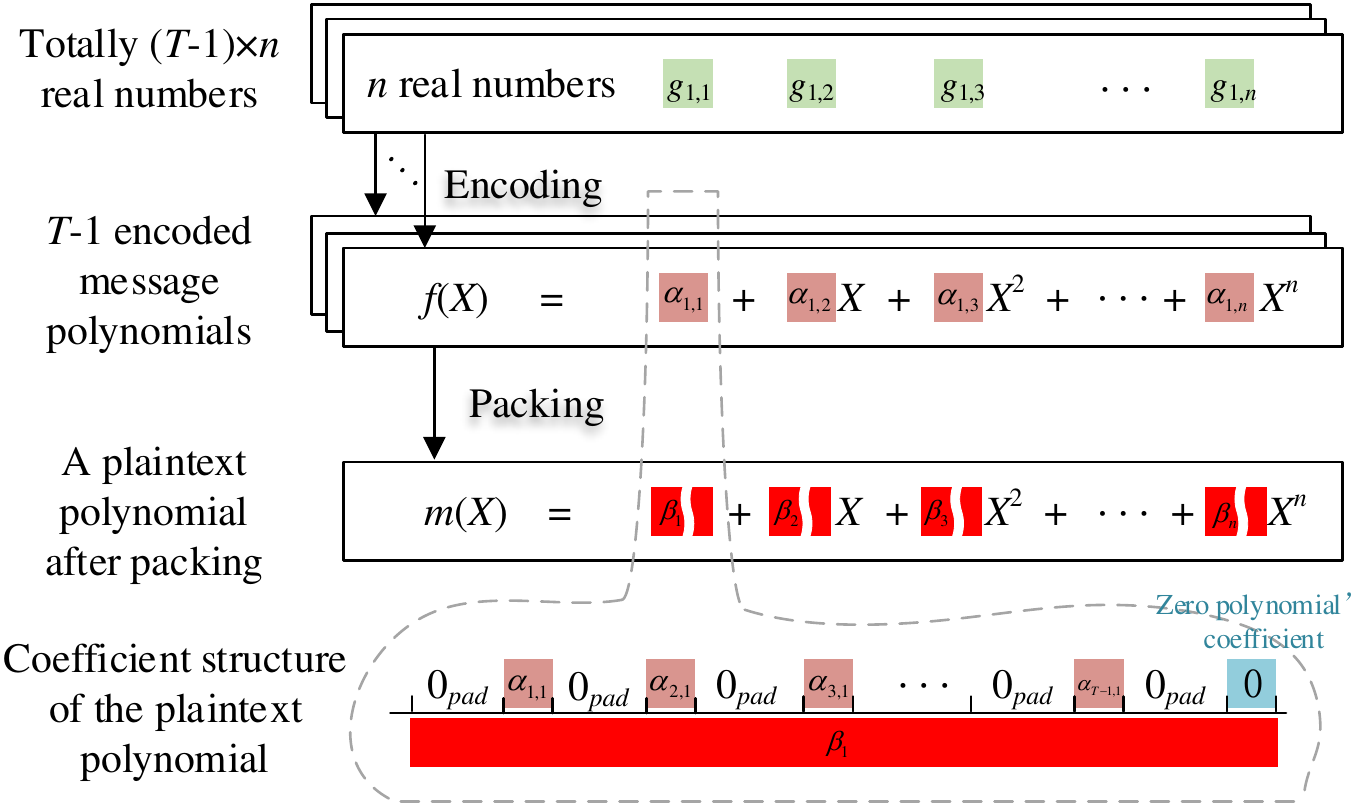}\\
  \caption{Messages encoding and polynomial packing. $(T-1)\times n$ real numbers are first encoded and then packed to form a plaintext polynomial.}\label{fig:packing}
\end{figure}

  \item[$\bullet$] $\texttt{EP.DcdUnpk}(m(X))\rightarrow \boldsymbol{g}'$:
     The unpacking and decoding algorithm \texttt{EP.DcdUnpk} is the reverse of \texttt{EP.EcdPack}.
     For each plaintext polynomial in the set $\{m(X)\in \mathcal{R}_{p}\}$, \texttt{EP.DcdUnpk} unpacks it as $T$ message polynomials $\{d_i(X)\in \mathcal{R}_{q_0\cdot 2^{pad}}\}_{i\in[T]}$, decodes the former $T-1$ message polynomials (the last message polynomial is the zero polynomial which is discarded) as $T-1$ complex vectors $\{\boldsymbol{d}_i\in \mathbb{C}^{n/2}\}_{i\in[T-1]}$, and rearranges the complex vectors as real messages.
      Specifically, \texttt{EP.DcdUnpk} unpacks and decodes each plaintext polynomial in the set $\{m(X)\}$ as follows.
      \begin{enumerate}[i), leftmargin=0.4cm]
        \item Polynomial unpacking: A polynomial $m(X)\in \mathcal{R}_{p}$ is unpacked as $T$ polynomials
            \begin{equation}
                d_i(X) = \sum_{j=0}^{n-1}\alpha_{i,j}X^j, \textrm{for } i\in[T]
            \end{equation}
             where $\alpha_{i,j}$ is calculated by shifting right $\beta_j$ (Eq.~\eqref{eq:packing}) by $(pad+\log q_0)(T+1-i)$ bits and intercepting the lowest $(pad+\log q_0)$ bits, i.e.,
            \begin{footnotesize}
            \begin{equation*}
                \alpha_{i,j}=
                \left[\beta_j >> (pad+\log q_0)(T+1-i)\right]_{\beta_j >> (pad+\log q_0)(T+2-i)},
            \end{equation*}
            \end{footnotesize}
            where ``$>>$" is a right shift operator that shifts the first operand right by the specified number (the second operand) of bits.

        \item Polynomial decoding: The $T-1$ polynomials $\{d_i(X)\}_{i\in[T-1]}$ are decoded as $T-1$ complex vectors 
              \begin{equation}
                \boldsymbol{d}_i = \phi(d_i(X))/p, \textrm{ for } i\in[T-1].
              \end{equation}
        After unpacking and decoding all the plaintext polynomials in $\{m(X)\}$, a message vector $\boldsymbol{g}'\approx \boldsymbol{g}$ can be obtained by rearranging all the complex vectors.
      \end{enumerate}

  \item[$\bullet$]
    $\texttt{EP.Eval}\left(\{m_i(X)\}_{i\in[N]} \rightarrow m_{\textrm{add}}(X)\right)$:
    Evaluation algorithm \texttt{EP.Eval} adds $N$ encoded and packed polynomials $\{m_i(X)\in \mathcal{R}_p\}_{i\in[N]}$ to an aggregation polynomials
    \begin{equation}
      m_{\textrm{add}}(X) = \sum_{i=1}^N m_i(X) = \sum_{i=1}^N
      \left( \sum_{j=0}^{n-1}\left(\beta_{i,j}X^j \right)\right) \in \mathcal{R}_p.
    \end{equation}
\end{itemize}

With the messages encoding and polynomials packing mechanism, a plaintext can contain up to $nT$ real numbers.



\begin{fig*}[!ht]
\begin{itemize}[leftmargin=0.3cm]
  \item[$\bullet$] \underline{Initialization}: On request from the $N$ clients, AS generates public parameters and secret keys. Alternatively, without AS assumption, the clients jointly generate public parameters and secret keys using MPC techniques.
\begin{enumerate}[leftmargin=0.5cm]
  \item \textbf{Parameter generation.}
      \begin{enumerate}[-]
        \item Execute \texttt{HE.Setup} and \texttt{EP.Setup} algorithms to generate the public parameters $\boldsymbol{PP}=\{1^{\lambda},p,q,\chi_s,\chi_e\}$, a random polynomials $a^0\in\mathcal{R}_q$, a random secret parameter $B$, the padding number $pad$, and the slot number $T$.
        \item Generate an initial model parameter $\boldsymbol{w}^0$.
        \item Determine the aggregation parameters $\{\alpha_i=|D_i|\}_{i\in[N]}$.
        \item Send all the clients the above parameters.
      \end{enumerate}

  \item \textbf{Key generation.}
      \begin{enumerate}[-]
        \item Execute \texttt{HE.KeyGen} algorithm to generate $N$ encryption keys $\{s_i\leftarrow\chi_s\}_{i\in[N]}$.
        \item Calculate the decryption key $s=\left[\sum_{i=1}^N s_i\right]_q$.
        \item Send the client $i$ with the keys $s_i$ and $s$.
      \end{enumerate}
\end{enumerate}

\item[$\bullet$] \underline{Client side (Training \& Encryption)}: Each client $i$ $(i\in[N])$ locally trains its model and executes \texttt{EP.EcdPack} and \texttt{MR.REnc} algorithms to generate its local gradient ciphertext.

\begin{enumerate}[leftmargin=0.5cm]
  \item \textbf{Initializing.} Initialize the local model and the random polynomial in encryption as $\boldsymbol{w}_i^1\leftarrow\boldsymbol{w}^0$ and $a^1\leftarrow a^0$, respectively.

  \item \textbf{Training.}
    Randomly select a mini-batch samples from the local dataset and then train the local model $\boldsymbol{w}_i^t$ to generate the local gradient $\boldsymbol{g}_i^t$, where $t\geq 1$ represents the $t$th round of training. 
    Weight the local gradient with $\alpha_i$ as $\boldsymbol{g}_i^t:=\alpha_i\boldsymbol{g}_i^t$.

  \item \textbf{Encoding and packing.}
      \begin{enumerate}[-]
        \item Normalize the weighted local gradient $\boldsymbol{g}_i^t$ into real numbers in range $[0, 1]$.
        \item Encode the normalized real numbers into integral coefficient polynomials in $R_{q_0}$ and pack each $T$ polynomials into one single polynomial in $\mathcal{R}_p$ using \texttt{EP.EcdPack} algorithm. The encoded and packed gradient is
            $$\boldsymbol{g}_i^t(X) \leftarrow \texttt{EP.EcdPack}(q_0,T,pad,\boldsymbol{g}_i^t).$$
      \end{enumerate}

  \item \textbf{Encryption.}
      Encrypt the encoded and packed gradient $\boldsymbol{g}_i^t(X)$ with \texttt{MR.REnc} algorithm to generate the local gradient ciphertext
            $$c_i^t\leftarrow\texttt{MR.REnc}(a^t, s_i,\boldsymbol{g}_i^t(X)).$$

  \item \textbf{Submission.} Submit the local gradient ciphertext $c_i^t$ to the cloud server.
  \item \textbf{Random polynomial generation.} During submission, generate a pseudo-random polynomial $a^{t+1}$ using \texttt{MR.PRPG} algorithm.
  \end{enumerate}

\item[$\bullet$] \underline{Cloud Side (Aggregation)}: After receiving all the $N$ clients' submissions, the cloud executes \texttt{HE.Eval} algorithm to generate an aggregation ciphertext.

\begin{enumerate}[leftmargin=0.5cm]
  \item \textbf{Aggregation.}
      Aggregate all the received local gradient ciphertexts as an aggregation ciphertext $C^t =\sum_{i=1}^N c_i^t.$

  \item \textbf{Transmission.}
      Transmit the aggregation ciphertext $C^t$ to the $N$ clients.
\end{enumerate}

\item[$\bullet$] \underline{Client Side (Decryption \& Model Update)}: Each client $i$ $(i\in[N])$ locally executes \texttt{EP.DcdUnpk} and \texttt{MR.RDec} algorithms to obtain the plaintext of the global gradient, and then updates its local model with the global gradient.

\begin{enumerate}[leftmargin=0.5cm]

  \item \textbf{Decryption.}
      Decrypt the aggregation ciphertext using \texttt{MR.RDec} algorithm to obtain the aggregation plaintext
        $$\boldsymbol{G'}^t\leftarrow \texttt{MR.RDec}(a^t, s, C^t).$$ 

  \item \textbf{Unpacking and decoding.}
      Unpack and decode the aggregation plaintext using \texttt{EP.DcdUnpk} algorithm to obtain the real parameter vector 
      $$\boldsymbol{G''}^t \leftarrow \texttt{EP.DcdUnpk}(\boldsymbol{G'}^t), j=1,2,\cdots$$

  \item \textbf{Parameter update.} Update the global model as $\boldsymbol{w}_i^{t+1}:=\boldsymbol{w}_i^t-\eta \boldsymbol{G''}^t/N$.

  \item \textbf{Training.} With the updated model parameter, each client continues to generate a local gradient for the next iteration until the final training model is obtained.
\end{enumerate}
\end{itemize}
 \setcounter{fig}{\value{figure}}
 \caption{Detailed description of the efficient and secure homomorphic encryption for federated learning (\nameScheme).}
 \label{fig:trainModel}
\end{fig*}

\section{\nameScheme}\label{sec:wholeModel}
We construct \nameScheme using \nameHE. The workflow of \nameScheme is summarized in Fig.~\ref{fig:trainModel}.

\textbf{Initialization.}
In this phase, public parameters and keys of \nameHE and the machine learning model are generated.
This generation process can be performed by the trusted AS or by the clients themselves.
Specifically, (1) Performing by AS:
In response to the clients' request, AS generates public parameters and encryption keys/decryption keys. 
Then, AS sends all the initialization parameters to the clients and distributes different encryption keys and the same decryption key to the clients. 
(2) Performing by the clients:
Without the trusted AS, the clients generate their respective encryption keys $s_i$ locally and then jointly generate the decryption key $s=\left[\sum_{i=1}^N s_i\right]_q$ using existing MPC techniques (such as the MPC technique in \cite{ben2008fairplaymp}), and the clients generate and distribute public parameters by one elected client. 

\textbf{Training and encryption.}
After initialization, each client trains its local model with its dataset, executes encoding, packing, and encryption algorithms to generate a local gradient ciphertext, and transmits the ciphertext to the cloud.
Then, each client generates a pseudo-random polynomial using a pseudo-random function for encryption in the next training round.

\textbf{Aggregation.}
The cloud aggregates all the received ciphertexts from the $N$ clients and sends the aggregation result/a global gradient ciphertext to the clients.

\textbf{Decryption and model update.} On receiving the global gradient ciphertext, each client executes decryption, unpacking, and decoding algorithms with its decryption key to obtain the global gradient plaintext.
The client then updates its local model and continues to train its model for the next training round until the final model is obtained.
%

\section{Theoretical Analysis}\label{sec:Analysis}
 In this section, we analyze the security and the complexity of \nameScheme.

\subsection{Security Analysis} \label{ssec:secAna}

The security of \nameScheme is ensured by \nameHE algorithm in which the clients' encryption/decryption keys are generated by a trusted party, the clients' submissions are encrypted using \nameHE, and aggregation results can only be decrypted locally by clients with the decryption key.
The security of \nameHE is ensured by the hardness of the RLWE assumption and the security/randomness of the pseudorandom generator \texttt{PRG} \cite{bonawitz2017practical}.
We next prove that \nameHE is IND-CPA secure in Theorem~\ref{the:IND-CPA}, \nameHE can provide multi-round privacy guarantees in Theorem~\ref{the:SecMultiRound}, and \nameHE is secure against collusion attacks of $\kappa<N-1$ clients in Theorem~\ref{the:SecCollusion}.

\begin{theorem}\label{the:IND-CPA}
    \nameHE is IND-CPA secure under the RLWE assumption of parameter $(n,q, \chi_s,\chi_e)$ and the security/randomness of \texttt{PRG}.
\end{theorem}

\begin{proof}
    Given a plaintext $m_i^t\in\mathcal{R}_p$, we define the distribution $\mathcal{D}=\{(B,a^0,a^t, \boldsymbol{c}_i^t)\}$ over $[0,1]^k\times\mathcal{R}_q^3$ as follows:\par
    \begin{enumerate}[(a)]
        \item $B\leftarrow U([0,1]^k),a^0\leftarrow U(\mathcal{R}_q),a^t\leftarrow\texttt{MR.PRPG}(t,B)=\texttt{Vec2Poly}\left(\texttt{PRG}\left((x_1 \oplus b_1)||\cdots|| (x_k\oplus b_k)\right)\right)$,
        \item $s_i\leftarrow\chi_s,e_i^t\leftarrow\chi_e,\boldsymbol{c}_i^t\leftarrow \texttt{MR.REnc}(a^t,s_i, m_i^t) = \left[a^t\cdot s_i+\tilde{p}\cdot e_i^t+m_i^t\right]_q$.
    \end{enumerate}

    We consider a distribution $\mathcal{D}'$ over $[0,1]^k\times\mathcal{R}_q^3$ which is obtained from $\mathcal{D}$ by changing its definition as
    \begin{enumerate}[(a)']
        \item $B\leftarrow U([0,1]^k),a^0\leftarrow U(\mathcal{R}_q),a^t\leftarrow U(\mathcal{R}_q)$.
    \end{enumerate}
    In definition (a), \texttt{MR.PRPG} is secure and its output is computationally indistinguishable from a uniformly sampled element of its output space $\mathcal{R}_q$ because its underlying pseudorandom generator \texttt{PRG} is secure and pseudo-random \cite{bonawitz2017practical}.
    Thus, the distributions $\mathcal{D}$ and $\mathcal{D}'$ are computationally indistinguishable.

    We next modify the definition (b) as
    \begin{enumerate}[(b)']
        \item $\boldsymbol{c}_i^t\leftarrow U(\mathcal{R}_q)$,
    \end{enumerate}
which is computationally indistinguishable from $\mathcal{D}'$ since the RLWE problem (related to key $s_i$) with parameter $(n,q, \chi_s,\chi_e)$ is hard.

Since $U([0,1]^k\times\mathcal{R}_q^3)$ is independent of the given plaintext $m_i^t$, we conclude that \nameHE is IND-CPA secure.
\end{proof}

\begin{theorem}\label{the:SecMultiRound}
    \nameHE can provide multi-round privacy guarantees under the RLWE assumption of parameter $(n,q, \chi_s,\chi_e)$ and the randomness of the primary random element $a^t, (t=1,2,\cdots)$.
\end{theorem}

\begin{proof}
Denote the $t$th and the $\tau$th round plaintexts of $N$ clients as $\{m_k^t\}_{k\in[N]}\in \mathcal{R}_p^N$ and $\{m_k^{\tau}\}_{k\in[N]}\in \mathcal{R}_p^N$, respectively, where $t=1,2,\cdots$, $\tau=1,2,\cdots$, and $t\neq\tau$.
Given a spanning-rounds aggregation ciphertext $\tilde{\boldsymbol{C}}^t=[\boldsymbol{C}^t-\boldsymbol{c}_i^t+\boldsymbol{c}_i^{\tau}]_q(i\in[N])$,
we define the distribution $\mathcal{D}=\{(a^t,a^{\tau},\tilde{\boldsymbol{C}}^t,\tilde{M}^t)\}$ over $\mathcal{R}_q^3\times\mathcal{R}_p$ as follows:
    \begin{enumerate}[(a)]
        \item $a^t\leftarrow\texttt{MR.PRPG}(t,B),a^{\tau}\leftarrow\texttt{MR.PRPG}(\tau,B)$,

        \item $s_k\leftarrow\chi_s,s_i\leftarrow\chi_s,e^t_k\leftarrow\chi_e,e_i^{\tau}\leftarrow\chi_e$,\\
        $\tilde{\boldsymbol{C}}^t=[\boldsymbol{C}^t-\boldsymbol{c}_i^t+\boldsymbol{c}_i^{\tau}]_q \leftarrow [\texttt{MR.REnc}(a^{\tau},s_i, m_i^{\tau}) +\sum_{k=1,k\neq i}^N \texttt{MR.REnc}(a^t,s_k, m_k^t)]_q = [(a^{\tau}\cdot s_i+\tilde{p}\cdot e_i^{\tau}+m_i^{\tau})+(a^t\sum_{k=1,k\neq i}^Ns_k+\tilde{p}\sum_{k=1,k\neq i}^Ne_k^t+\sum_{k=1,k\neq i}^Nm_k^t)]_q$, 

        \item $\tilde{M}^t \leftarrow \texttt{MR.RDec}(a^t, s,\tilde{\boldsymbol{C}}^t) = [[\tilde{\boldsymbol{C}}^t-a^t\cdot s]_q]_p$.
    \end{enumerate}

    We consider a distribution $\mathcal{D}'$ over $\mathcal{R}_q^3\times\mathcal{R}_p$ which is obtained from $\mathcal{D}$ by changing its definition as
    \begin{enumerate}[(a)']
        \item $a^t\leftarrow U(\mathcal{R}_q), a^{\tau}\leftarrow U(\mathcal{R}_q)$,
        \item $\tilde{\boldsymbol{C}}^t\leftarrow U(\mathcal{R}_q)$.
    \end{enumerate}
    The distributions $\mathcal{D}$ and $\mathcal{D}'$ are computationally indistinguishable because of
    the randomness of the primary random element $a^t$ and $a^{\tau}$ in definition (a) and the randomness of the RLWE samples in definition (b).

    We then modify definition (c) as
    \begin{enumerate}[(c)']
        \item $\tilde{M}^t \leftarrow U(\mathcal{R}_p)$.
    \end{enumerate}
    In definition (c), we have
    \begin{equation*}
    \footnotesize
        \begin{split}
              & \tilde{M}^t \leftarrow \texttt{MR.RDec}(a^t, s,\tilde{\boldsymbol{C}}^t)
            = \left[\left[\tilde{\boldsymbol{C}}^t-a^t\cdot s\right]_q\right]_p \\
            = & \left[\left[(a^t+a^{\tau})s_i + \tilde{p}\left(\sum_{k=1,k\neq i}^Ne_k^t+e_i^{\tau}\right) + \sum_{k=1,k\neq i}^Nm_k^t+m_i^{\tau} \right]_q\right]_p, 
        \end{split}
    \end{equation*}
    where ${\scriptstyle \left[(a^t+a^{\tau})s_i + \tilde{p}\left(\sum_{k=1,k\neq i}^Ne_k^t+e_i^{\tau}\right) + \sum_{k=1,k\neq i}^Nm_k^t+m_i^{\tau} \right]_q}$ is computationally indistinguishable from a random sample of $\mathcal{R}_q$ due to the hardness of the RLWE problem.
    Thus $\tilde{M}^t$ is computationally indistinguishable from a random sample of $\mathcal{R}_p$ and the distributions $\mathcal{D}$ and $\mathcal{D}'$ are computationally indistinguishable.

Since $U(\mathcal{R}_q^3\times\mathcal{R}_p)$ is independent of the given spanning-rounds aggregation ciphertext and the corresponding $N$ original plaintexts in two rounds, i.e., the $t$th and the $\tau$th round, \nameHE can provide privacy guarantees of two-round encryption.
Similarly,  it can be proven that $U(\mathcal{R}_q^3\times\mathcal{R}_p)$ is independent of the given spanning-rounds aggregation ciphertext involving arbitrary multiple rounds.
Therefore, we conclude that \nameHE can provide multi-round privacy guarantees.
\end{proof}

\begin{theorem}\label{the:sumOfRandomVar}
  Let $X$ and $Y$ be independent random variables that are normally distributed, then their sum is also normally distributed. i.e., if $X\sim N(\mu_X,\sigma^2_X)$ and $Y\sim N(\mu_Y,\sigma^2_Y)$ then $Z\sim N(\mu_X+\mu_Y,\sigma^2_X+\sigma^2_Y)$, where $Z=X+Y$ is the sum of the two variables and $N(\mu,\sigma^2)$ represents normal (Gaussian) distribution with mean $\mu$ and variance $\sigma^2$.
\end{theorem}
The proof of Theorem~\ref{the:sumOfRandomVar} can be seen in \cite{wikipediaGaussian}.

\begin{theorem}\label{the:randomnessOfCondition}
Let $X_1, X_2,\cdots$ be independent random variables that obey normal (Gaussian) distribution $N(\mu,\sigma^2)$.
Denote the values of the random variables $X_1,X_2,\cdots$ as $x_1,x_2,\cdots$, respectively.
Then, for all given constant values $y=x_1+x_2+\cdots$, the possible value of $x_i$ follows Gaussian distribution.
\end{theorem}


\begin{proof}
  Let $Y=\sum_{i=1}^n X_i$.
  We give a proof by mathematical induction on the natural number $n$.\par
  The probability density function of the variable $X_i$ is $f_{X_i}(x_i)=\frac{1}{\sqrt{2\pi}\sigma}\exp\left(-\frac{(x_i-\mu)^2}{2\sigma^2}\right)$.
  According to Theorem~\ref{the:sumOfRandomVar}, the probability density function of $Y$ is $f_{Y}(y)=\frac{1}{\sqrt{2n\pi}\sigma}\exp\left(-\frac{(x_i-n\mu)^2}{2n\sigma^2}\right)$.
  For all given constant value $y$, the probability density function for possible values of $x_i$ is $f_{X_i|Y}(x_i|y)$ which is the conditional probability density function of $X_i$ given $Y$.\\
  \textit{Base case:} Show that Theorem~\ref{the:randomnessOfCondition} holds for $n=2$.\par
  When $n=2$, we have $Y=X_1+X_2$. 
  The joint probability density function of $X_1$ and $Y$ is $f(x_1,y)=f_{X_1}(x_1)\cdot f_{X_2}(y-x_1)$.
  Then, the conditional probability density function of $X_1$ given $Y$ is
\begin{equation}
\begin{split}
  &\ f_{X_1|Y}(x_1|y)=\frac{f(x_1,y)}{f_Y(y)}=\frac{f_{X_1}(x_1)\cdot f_{X_2}(y-x_1)}{f_Y(y)}\\
  &=\frac{\frac{1}{2\pi\sigma^2} \exp\left(-\frac{(x_1-\mu)^2+(y-x_1-\mu)^2}{2\sigma^2}\right)}{\frac{1}{2\sqrt{\pi}\sigma}\exp\left(-\frac{(y-2\mu)^2}{4\sigma^2}\right)}\\
  &\triangleq A_1\cdot \frac{1}{\sqrt{2\pi}\sigma_1}\exp\left(-\frac{(x_1-\mu_1)^2}{2\sigma_1^2}\right),
\end{split}
\end{equation}
where $A_1=\exp\left(-\frac{y^2}{2\sigma^2}\right),\mu_1=\frac{y}{2}, \sigma_1^2=\frac{\sigma^2}{2}$.
For a fixed $y$, $f_{X_1|Y}(x_1|y)$ can be seen as the result of multiplying the probability density function of Gaussian distribution $N(\mu_1,\sigma_1^2)$ by the scaling factor $A_1$.
That is to say, for all given constant $y=x_1+x_2$, the possible value of $x_i (i=1,2)$ follows Gaussian distribution.\\
\textit{Induction step:} Show that for every $k\geq 2$ and $y=\sum_{i=1}^k x_i$, if Theorem~\ref{the:randomnessOfCondition} holds, then for $y=\sum_{i=1}^{k+1} x_i$, Theorem~\ref{the:randomnessOfCondition} also holds.\par
Assume the induction hypothesis that for a particular $k$, the single case $n=k$ holds, i.e.,
for all fixed $y$, $f_{X_{k}|Y}(x_{k}|y)$ can be seen as the result of multiplying the probability density function of Gaussian distribution $N(\mu_k,\sigma_k^2)$ by the scaling factor $A_k$.

When $n=k+1$, we have $Y=X_1+\cdots+X_{k+1}$. The joint probability density function of $X_{k+1}$ and $Y$ is $f(x_{k+1},y)=f_{X_{k+1}}(x_{k+1})f(x_{k},y-x_{k+1})$.
Then, the conditional probability density function of $X_{k+1}$ given $Y$ is
\begin{small}
\begin{align*}
  f_{X_{k+1}|Y}&(x_{k+1}|y)=\frac{f(x_{k+1},y)}{f_Y(y)}
  = \frac{f_{X_{k+1}}(x_{k+1})\cdot f(x_{k},y-x_{k+1})}{f_Y(y)}\\
  =&\ \frac{f_{X_{k+1}}(x_{k+1})\cdot f_{X_{k}|Y}(x_{k}|(y-x_{k+1}))\cdot f_Y(y-x_{k+1})}{f_Y(y)}\\
  =&\ \frac{A_{k}}{2\pi\sigma\sigma_k} \exp\left(-\frac{(x_{k+1}-\mu)^2}{2\sigma^2}-\frac{(y-x_{k+1}-\mu_k)^2}{2\sigma_k^2}\right)\\
  &\ \cdot \exp\left(-\frac{x_{k+1}^2-2yx_{k+1}+4\mu x_{k+1}}{4\sigma^2}\right)\\
  \triangleq &\ A_{k+1}\cdot\frac{1}{\sqrt{2\pi}\cdot \sigma_{k+1}}\exp\left(-\frac{\left(x_{k+1}-\mu_{k+1}\right)^2}{2\sigma_{k+1}^2}\right),
\end{align*}
\end{small}
where ${\scriptstyle A_{k+1}=\frac{A_{k}\cdot\exp\left(-\frac{2\sigma_k^2\mu^2+2\sigma^2(\mu_k^2+y^2-2y\mu_k)-(\sigma_k^2y+2\sigma^2y-2\sigma^2\mu_k)^2}{4\sigma^2\sigma_k^2}\right)}{\sqrt{\pi(3\sigma_k^2+2\sigma^2)}}}$,
${\scriptstyle\mu_{k+1}=\frac{\sigma_k^2y+2\sigma^2y-2\sigma^2\mu_k}{3\sigma_k^2+2\sigma^2}}$, and  ${\scriptstyle\sigma_{k+1}^2=\frac{2\sigma^2\sigma_k^2}{3\sigma_k^2+2\sigma^2}}$.\par
For a fixed $y$, $f_{X_{k+1}|Y}(x_{k+1}|y)$ can be seen as the result of multiplying the probability density function of Gaussian distribution $N(\mu_{k+1},\sigma_{k+1}^2)$ by the scaling factor $A_{k+1}$.
%
Thus, for all given constant $y=x_1+\cdots+x_{k+1}$, the possible value of $x_{k+1}$ follows Gaussian distribution.
That is, when $n=k+1$, Theorem~\ref{the:randomnessOfCondition} also holds true, establishing the induction step.\\
\textit{Conclusion:} Since both the base case and the induction step have been proven as true, by mathematical induction Theorem~\ref{the:randomnessOfCondition} holds for every natural number $n$.
\end{proof}

\begin{theorem}\label{the:SecCollusion}
    \nameHE is secure against collusion attacks of $\kappa<N-1$ clients. 
\end{theorem}

\begin{proof}
    Given the $N-2$ colluded clients' plaintexts $\{m_k\}_{k\in[N-2]}\in\mathcal{R}_p^{N-2}$ and the decryption key $s\in\mathcal{R}$,
    we define the distribution $\mathcal{D}=\{(s_i,e_i,m_i), i\in\{N-1,N\}\}$ over $\mathcal{R}^2\times\mathcal{R}_p$ as follows:\par
    \begin{enumerate}[(a)]
        \item $s_j\leftarrow\chi_s (j\in[N]),s_{N-1}+s_N = s - \sum_{k=1}^{N-2}s_k\triangleq s'$,
        \item $a\leftarrow U(\mathcal{R}_p),e_j\leftarrow\chi_e$, \\
        $\boldsymbol{c}_j\leftarrow \texttt{MR.REnc}(a,s_j, m_j) = \left[a\cdot s_j+\tilde{p}\cdot e_j+m_j\right]_q$, \\
        $\boldsymbol{C}\leftarrow\texttt{MR.REval}\left(\{\boldsymbol{c}_j\}_{j\in[N]}\right)=\sum_{j=1}^N\boldsymbol{c}_j (j\in[N])$,\\
        $M\leftarrow\texttt{MR.RDec}(a,s,\boldsymbol{C})=[[\boldsymbol{C}-a\cdot s]_q]_p$,
        $e_{N-1}+e_N =[\boldsymbol{c}_{N-1}+\boldsymbol{c}_N - a(s_{N-1}+s_N) - (M - \sum_{k=1}^{N-2}m_k)]_q|p\triangleq e'$,
        \item $m_{N-1}+m_N = M - \sum_{k=1}^{N-2}m_k$.
    \end{enumerate}

    We consider a distribution $\mathcal{D}'$ over $\mathcal{R}^2\times\mathcal{R}_p$ which is obtained from $\mathcal{D}$ by changing its definition as
    \begin{enumerate}[(a)']
        \item $s_i\leftarrow\chi_s' (i\in\{N-1,N\})$,
        \item $e_i\leftarrow\chi_e' (i\in\{N-1,N\})$,
        \item $m_i\leftarrow U(\mathcal{R}_p) (i\in\{N-1,N\})$.
    \end{enumerate}
    In definitions (a) and (b), according to Theorem~\ref{the:randomnessOfCondition},
    for a fixed $s'$ and a fixed $e'$,
    $s_i$ and $e_i$ $(i\in\{N-1,N\})$ obey Gaussian distributions which we denote as $\chi_s'$ and $\chi_e'$, respectively.
    In definition (c), the colluded clients know only the honest clients' ciphertexts $\boldsymbol{c}_{N-1}\leftarrow \texttt{MR.REnc}(a,s_{N-1}, m_{N-1}) = \left[a\cdot s_{N-1}+\tilde{p}\cdot e_{N-1}+m_{N-1}\right]_q$ and
    $\boldsymbol{c}_N\leftarrow \texttt{MR.REnc}(a,s_N, m_N) = \left[a\cdot s_N+\tilde{p}\cdot e_N+m_N\right]_q$ and the sum of their plaintexts $m_{N-1}+m_N = M - \sum_{k=1}^{N-2}m_k$, so $m_i\leftarrow U(\mathcal{R}_p) (i\in\{N-1,N\})$ holds because of the hardness of the RLWE assumption.
    Thus, the distributions $\mathcal{D}$ and $\mathcal{D}'$ are computationally indistinguishable.

Since the distribution $\mathcal{D}$ is independent of the given plaintext $\{m_k\}_{k\in[N-2]}\in\mathcal{R}_p^{N-2}$, we conclude that \nameHE is secure against collusion attacks of $\kappa<N-1$ clients.
\end{proof}

\subsection{Complexity Analysis}\label{ssec:compAna}
Here, we analyze the complexity of \nameScheme and four state-of-the-art privacy-preserving cross-silo FL models using different HE algorithms in \cite{chen2019efficient,ma2022privacy, Zhang2020BatchCrypt, FATE2019}, where HE algorithms in \cite{chen2019efficient,ma2022privacy} are multi-key variants of CKKS cryptosystem with additive homomorphism, and HE algorithms in \cite{Zhang2020BatchCrypt, FATE2019} are single-key additively HE algorithms.

\textbf{Computational Complexity.}
 We denote the number of collaborative clients as $N$ and the number of the client-cloud iteration as $M$.
 For one client, the number of weight parameters in a network model is denoted as $\varrho$.
 The security parameter of encryption is denoted as $n$.

In the client-side training phase, each client mainly performs model training.
All five models have the same computational complexity for the $N$ clients, that is, $O(NM\varrho)$.

In the client-side encryption phase, each client mainly encrypts its local gradient.
In FL models \cite{Zhang2020BatchCrypt, FATE2019} and \nameScheme, their encryption of a plaintext $m$ is a ciphertext $c$.
That is, in one iteration, each client performs its encryption algorithm once.
Thus, the computational complexity of \cite{Zhang2020BatchCrypt, FATE2019} and \nameScheme is $O(NMn\varrho)$.
In \cite{chen2019efficient}, an MKHE encryption of a plaintext $m$ is a ciphertext vector $ct = (c_0,c_1,\cdots,c_N)$, where $c_0$ is the ciphertext of $m$ and $c_j$ $(j\in[N])$ is a ciphertext of a random number, and these ciphertexts are all encrypted by an RLWE-based encryption algorithm.
That is, in each iteration, each client should perform its encryption algorithm $N+1$ times.
Thus, the computational complexity of FL using MKHE in \cite{chen2019efficient} is $O((N+1)NMn\varrho)$.
In \cite{ma2022privacy}, an MKHE encryption of a plaintext $m$ is a ciphertext vector $ct = (c_0,c_1)$, where $c_0$ and $c_1$ are the ciphertexts of $m$ and a random number, respectively. 
Thus, the computational complexity of FL using MKHE in \cite{ma2022privacy} is $O(2NMn\varrho)$.

In the cloud-side aggregation phase, the cloud aggregates all $N$ clients' ciphertexts/ciphertext vectors as an aggregation ciphertext/ciphertext vector.
In FL models \cite{Zhang2020BatchCrypt, FATE2019} and \nameScheme, their aggregation processes are performed as $C=\sum_{i=1}^Nc_i$,
so their computational complexity is $O(NMn\varrho)$.
The aggregation process in \cite{chen2019efficient} is performed as $C=(\sum_{i=1}^N ct_i[0],\cdots,\sum_{i=1}^N ct_i[N])$, where $ct_i[j] (j\in[0,N])$ represents the $j+1$th element of the ciphertext vector $ct_i$.
Thus, the computational complexity of FL using MKHE in \cite{chen2019efficient} is $O((N+1)NMn\varrho)$.
The aggregation process in \cite{ma2022privacy} is performed as $C=(\sum_{i=1}^N ct_i[0],\sum_{i=1}^N ct_i[1])$.
Thus, the computational complexity of FL using MKHE in \cite{ma2022privacy} is $O(2NMn\varrho)$.

In the client-side decryption phase, each client decrypts aggregation results.
In FL models \cite{Zhang2020BatchCrypt, FATE2019} and \nameScheme, each client decrypts an aggregation ciphertext by performing its decryption algorithm once, and thus the computational complexity of a client is $O(Mn\varrho)$. The total computational complexity of decryption for $N$ clients is $O(NMn\varrho)$.
In \cite{chen2019efficient}, the client $i$ $(i\in[N])$ first partially decrypts the $(i+1)$th ciphertext of the aggregation ciphertext vector $C$ by using its own secret key, and then it merges all clients' partial decryption results as a final decrypted result.
Thus, the computational complexity of decryption for each client is $O(Mn\varrho+NMn\varrho)$ and the computational complexity of decryption for FL models using MKHE in \cite{chen2019efficient} is $O((N+1)NMn\varrho)$.
Similarly, in \cite{ma2022privacy},  the client $i$ $(i\in[N])$ first partially decrypts the second ciphertext of the aggregation ciphertext vector $C$, i.e., $\sum_{i=1}^N ct_i[1]$, by using its own secret key, and then it performs full decryption by merging all clients' partial decryption results.
Thus, the computational complexity of decryption for each client is $O(Mn\varrho+2Mn\varrho)$ and the computational complexity of decryption for FL models using MKHE in \cite{chen2019efficient,ma2022privacy} is $O(3NMn\varrho)$.


Therefore, the computational complexity of the whole FL models using MKHE in \cite{chen2019efficient} and \cite{ma2022privacy} is $O(3(N+1)NMn\varrho+NM\varrho)$ and $O(7NMn\varrho+NM\varrho)$, respectively, while the computational complexity of the whole FL models using SKHE in \cite{Zhang2020BatchCrypt, FATE2019} and the whole \nameScheme model is $O(3NMn\varrho+NM\varrho)$. 

 \textbf{Communication Complexity.}
 We denote the space cost of the network model as $\varpi$.
 In the client-side training phase, all $N$ clients send their encrypted gradients to the cloud.
 In \cite{Zhang2020BatchCrypt, FATE2019} and \nameScheme, their ciphertext numbers are all 1, while in \cite{chen2019efficient} and \cite{ma2022privacy}, their ciphertext numbers are $N+1$ and $2$, respectively.
 Thus, for FL models using SKHE in \cite{Zhang2020BatchCrypt, FATE2019} and \nameScheme, the communication complexity from the $N$ clients to the cloud server is $O(NMn\varpi)$, and for FL models using MKHE in \cite{chen2019efficient} and \cite{ma2022privacy}, the communication complexity from the $N$ clients to the cloud server is $O((N+1)NMn\varpi)$ and $O(2NMn\varpi)$, respectively. %

 In the cloud-side aggregation phase, the cloud aggregates all clients' ciphertexts/ciphertext vectors and transmits the aggregation result to all $N$ clients.
 The numbers of ciphertext of an aggregation result in \cite{Zhang2020BatchCrypt, FATE2019} and \nameScheme are all 1, while that in \cite{chen2019efficient} and \cite{ma2022privacy} are $N+1$ and 2, respectively.
 Thus, for FL models using SKHE in \cite{Zhang2020BatchCrypt, FATE2019} and \nameScheme, the communication complexity from the cloud server to the $N$ clients is $O(NMn\varpi)$, and for FL models using MKHE in \cite{chen2019efficient} and \cite{ma2022privacy}, the communication complexity from the cloud server to the $N$ clients is $O((N+1)NMn\varpi)$ and $O(2NMn\varpi)$, respectively.

 In the client-side decryption phase, for FL models using SKHE in \cite{Zhang2020BatchCrypt, FATE2019} and \nameScheme, each client decrypts aggregation results locally, so their communication complexities of decryption are all 0.
 For FL models using MKHE in \cite{chen2019efficient} and \cite{ma2022privacy}, each client first partially decrypts aggregation results and then sends its partial decryption results to all other $N-1$ clients for final decryption, so their communication complexities of decryption are $O((N-1)NMn\varpi)$.

 Therefore, the whole communication complexity of FL models using SKHE in \cite{Zhang2020BatchCrypt, FATE2019} and \nameScheme is $O(NMn\varpi)$, and the whole communication complexity of FL models using MKHE in \cite{chen2019efficient} and \cite{ma2022privacy} is $O(N^2Mn\varpi)$.

From the above complexity analyses, we know that \nameScheme has lower computation and communication complexity compared with FL models using MKHEs.
Besides, the computation and communication complexity of \nameScheme is the same as that of SKHE-based FL models, whose computation and communication complexity is linearly related to the number of clients, the number of iterations, the security parameter of encryption, and the network architecture.

\begin{figure*}[ht]
\hspace{-0.45cm}
	\subfloat[FCN test accuracy with time]{
		\begin{minipage}[b]{.33\textwidth}
			\centering
			\includegraphics[width=2.4in]{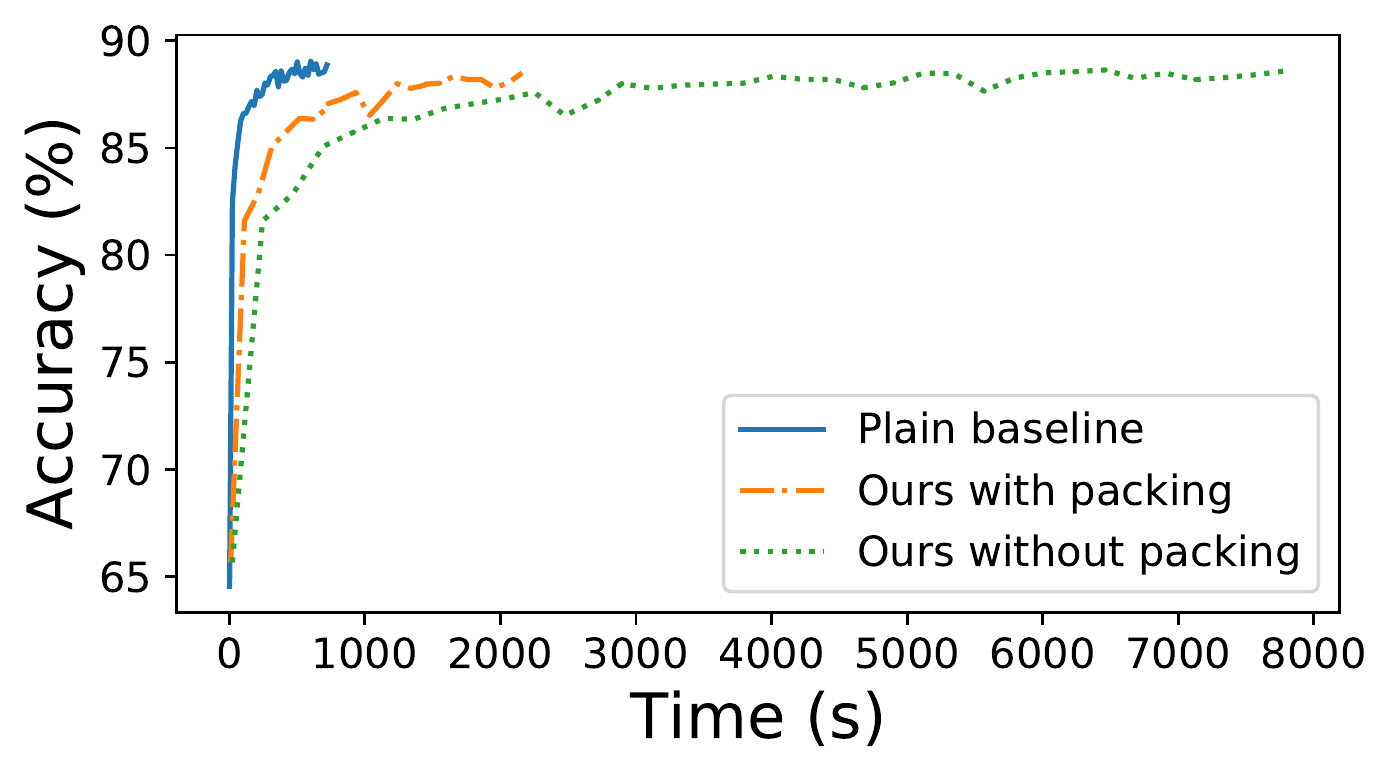}
		\end{minipage}
	}
	\subfloat[AlexNet test accuracy with time]{
		\begin{minipage}[b]{.33\textwidth}
			\centering
			\includegraphics[width=2.4in]{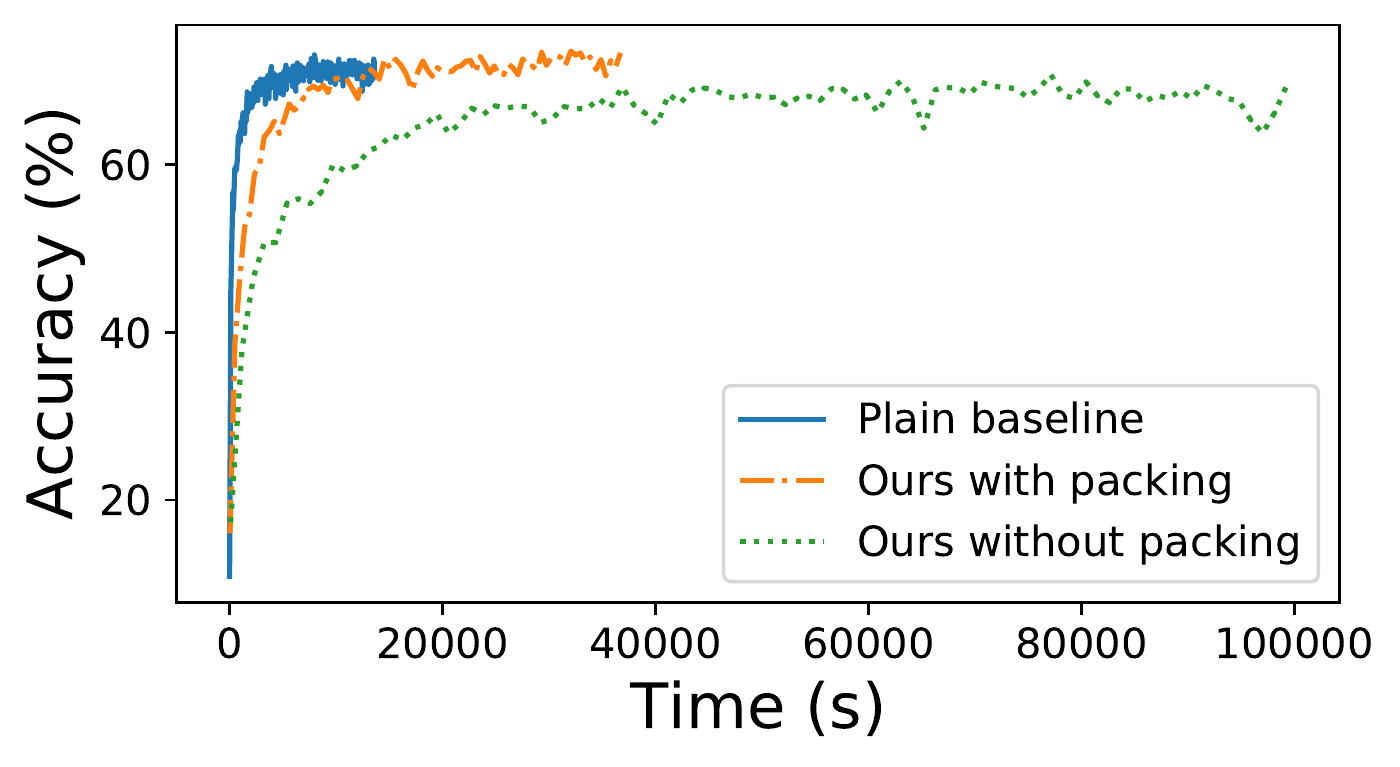}
		\end{minipage}
	}
	\subfloat[LSTM train loss with time]{
		\begin{minipage}[b]{.33\textwidth}
			\centering
			\includegraphics[width=2.6in]{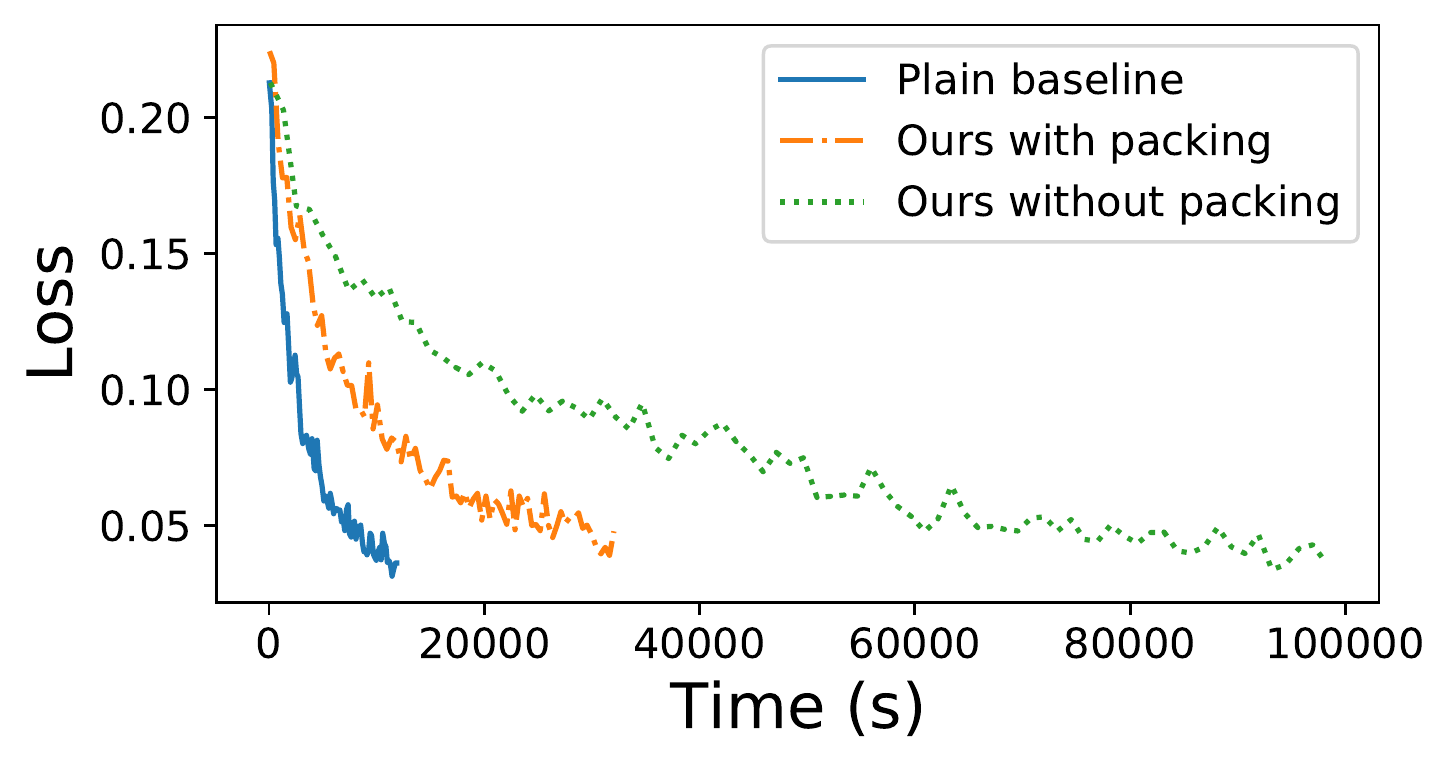}
		\end{minipage}
	}
\\
	\subfloat[FCN test accuracy with epoch]{
	    \hspace{-0.45cm}
		\begin{minipage}[b]{.33\textwidth}
			\includegraphics[width=2.4in]{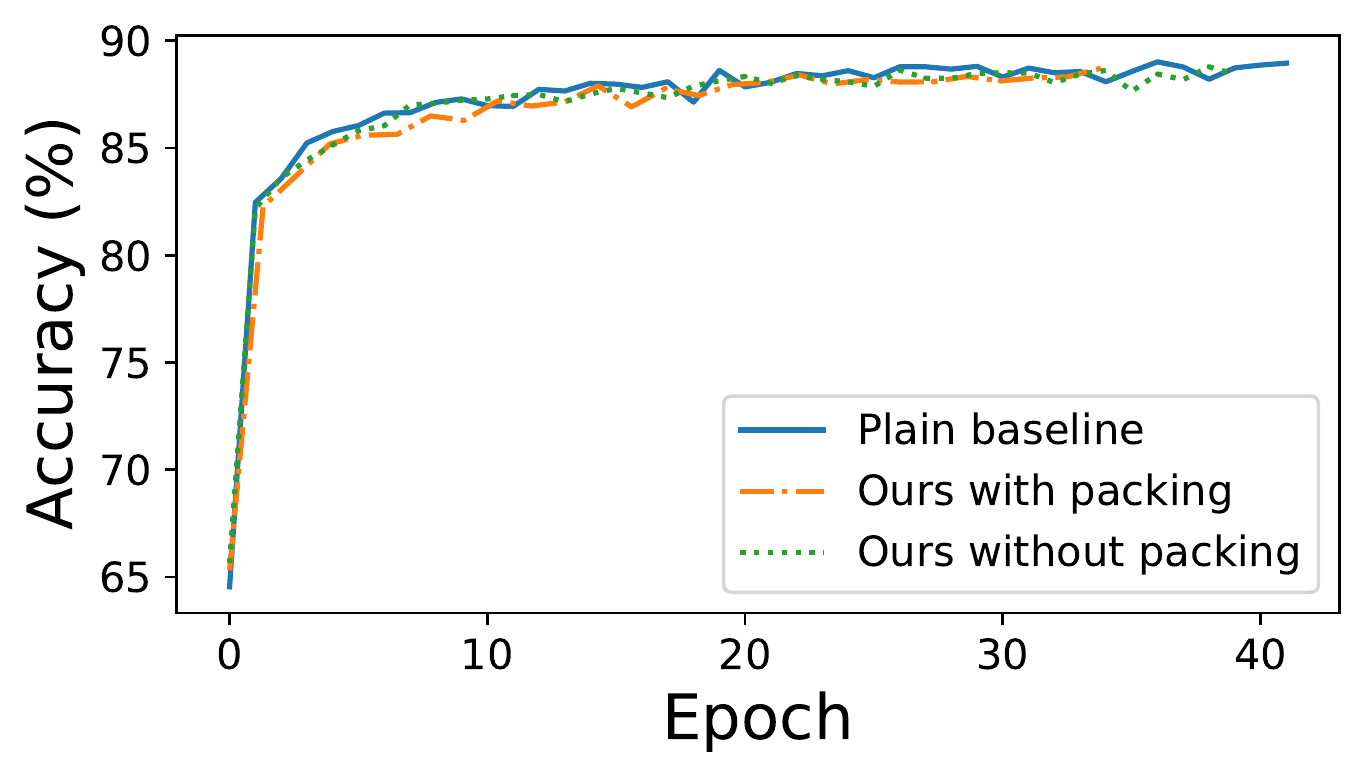}
		\end{minipage}
	}
	\subfloat[AlexNet test accuracy with epoch]{
		\begin{minipage}[b]{.33\textwidth}
			\includegraphics[width=2.4in]{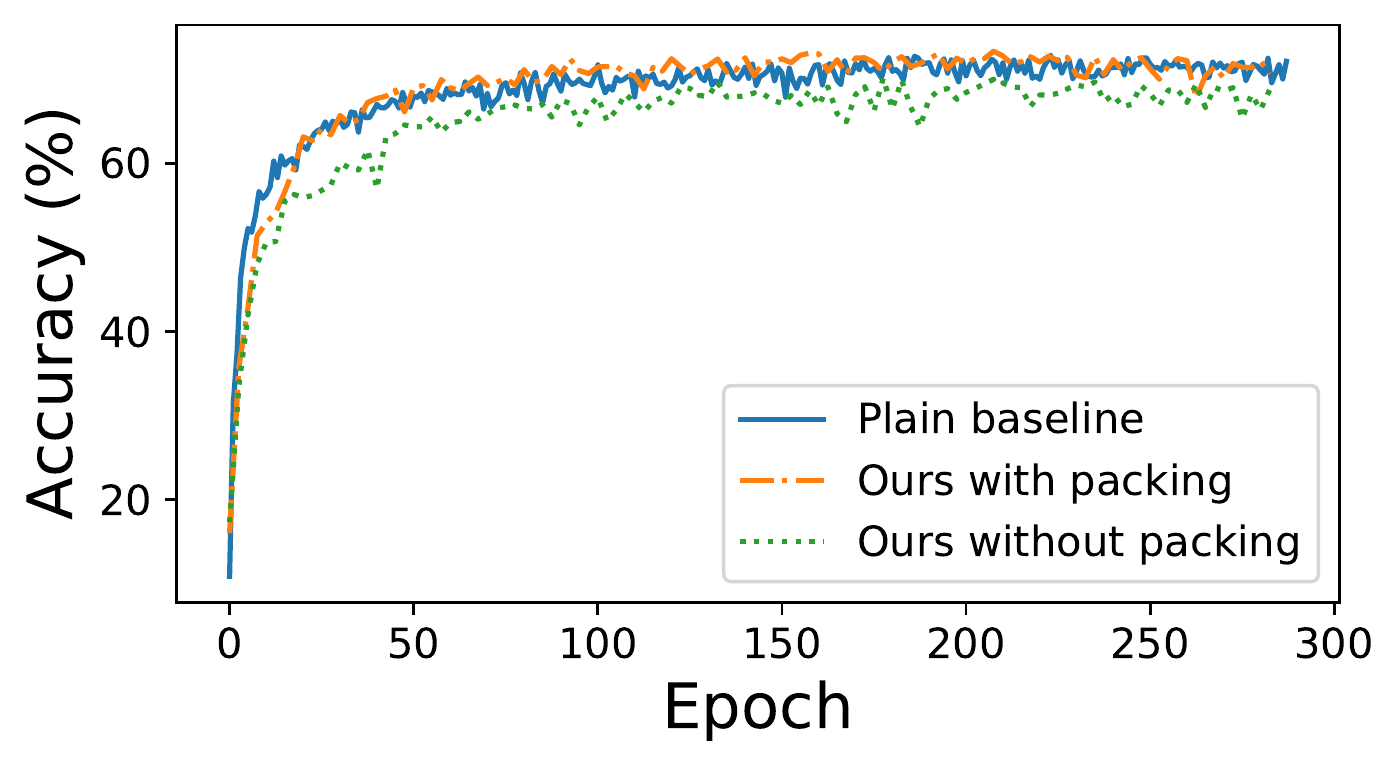}
		\end{minipage}
	}
	\subfloat[LSTM train loss with epoch]{
		\begin{minipage}[b]{.33\textwidth}
			\includegraphics[width=2.5in]{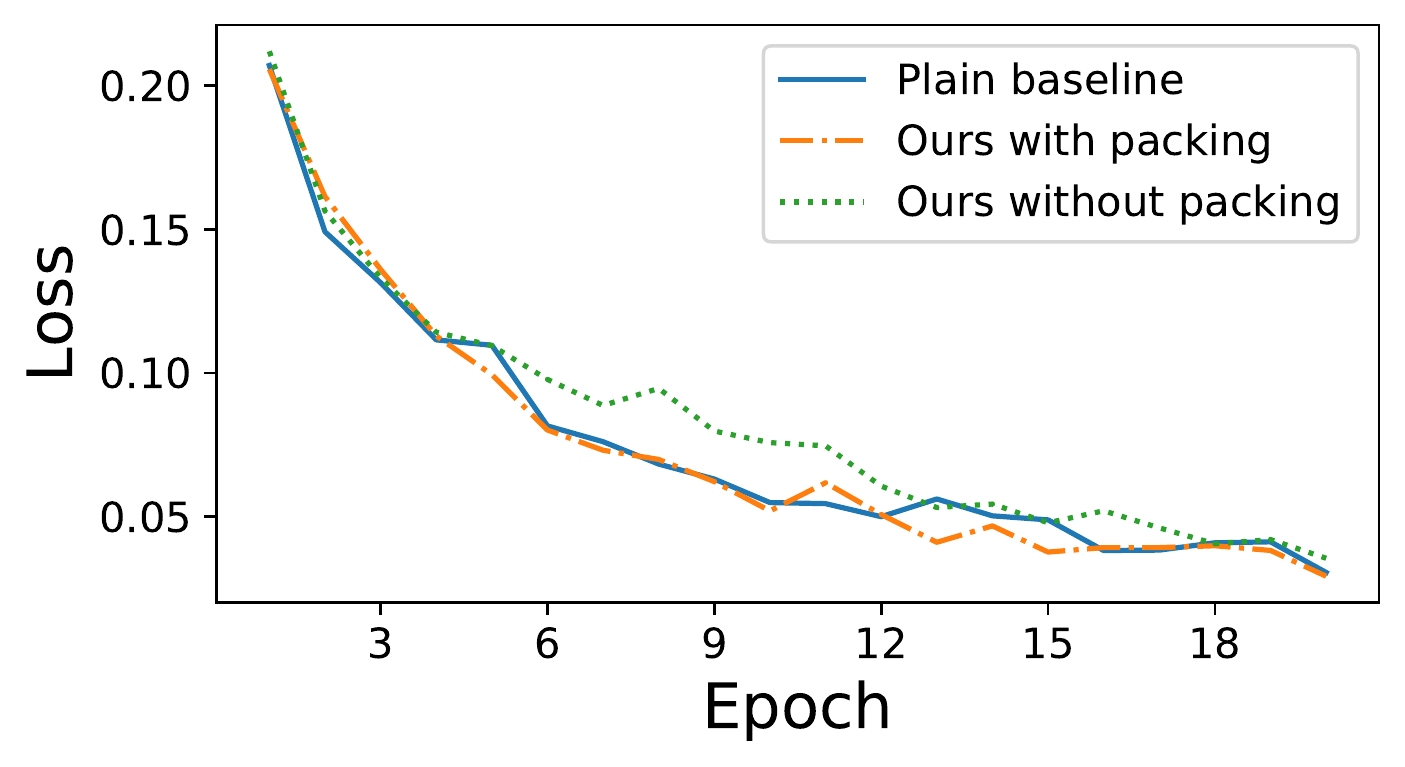}
		\end{minipage}
	}
    \setcounter{figure}{\value{fig}}	
	\caption{The quality and the training speed of plain baseline and our \nameScheme with and without packing}
	\label{fig:unpkVSpack}
\end{figure*}

\section{Performance Evaluation}\label{sec:experiments}

In this section, we evaluate the performance of \nameScheme via experiments.
We first examine the effectiveness of \nameScheme's packing (Section~\ref{ssec:impactPack}).
We then assess the computation and communication benefits \nameScheme brings compared to FL models using SKHE in \cite{Zhang2020BatchCrypt, FATE2019} and how \nameScheme performs compared to the baseline federated learning on plaintext (Section~\ref{ssec:efficiency}).
Finally, we demonstrate the total training time and communication traffic that FL models take until they converge.

\subsection{Experimental Setup}

\subsubsection{Implementation}\label{sssec:implemetation}

The experiments are done on a computer with 
Intel Core i7-6700 CPU.
The \nameHE's encryption/decryption, encoding/decoding, and packing/unpacking algorithms are coded with C++ language and employ NTL 10.4.0 \cite{ntl} and GMP 6.2.1 libraries \cite{gmp} for implementing arbitrary length integers and arbitrary precision arithmetic.
The federated learning models are performed with PyTorch 1.11.0 framework over Python 3.8. 
To implement \nameScheme, 
we generate a dynamic library of all the \nameHE codes and call it in a Python script. 

\subsubsection{Benchmark Models and \nameHE Parameters}

We denote FL models of \cite{Zhang2020BatchCrypt} and \cite{FATE2019} as BatchCrypt and FATE, respectively.
For the convenience of comparison, we implement three ML applications in both \nameCompOne and \nameCompTwo.
In the first application, a 3-layer fully-connected neural network (FCN) with 101.77K weights is trained over FMNIST dataset, with batch size 128 and Adam optimizer.
The second application is training an AlexNet model with 1.25M weights using CIFAR10 dataset, where training batch size is 128 and RMSprop optimizer is adopted with $10^{-6}$ decay.
The third application is an LSTM model with 4.02M weights trained over Shakespeare dataset, with batch size 64 and Adam optimizer.
There are nine clients in \nameScheme. For each application, we randomly shuffle its dataset and assign it to the nine clients.
The clients' datasets in FCN and AlexNet are IID (identically and independently distributed) while those in LSTM are non-IID.
We set $\alpha_1=\cdots=\alpha_9=1$.

\begin{table}
\centering
\caption{Network bandwidth between the server and the clients in different regions.}
\label{tab:bandwidth}
\begin{tabular}{|l|l|l|l|l|l|}
\hline
Region & HK &LDN  & TYO. & N. VA.  & Ore. \\\hline
Uplink (Mbps) &81 &97  &116  &165  & 9841\\\hline
Downlink (Mbps) &84 &84 &122 &151 & 9841 \\\hline
\end{tabular}
\end{table}

We consider the same geo-distributed FL scenario as \nameCompOne, where the clients train an ML model in five AWS EC2 datacenters located in Hong Kong, London, Tokyo, N. Virginia, and Oregon, and the aggregator is run in the Oregon datacenter.
The network bandwidth between the aggregator and the clients is presented in Table~\ref{tab:bandwidth}.
We employ synchronous stochastic gradient descent method to train all the ML models.

The precision of all three models is $\log q_0=16$ unless otherwise specified.
We choose parameters of \nameHE according to Albrecht's LWE estimator \cite {Albrecht2021} which estimates the bit security of the RLWE-based encryption algorithm on known attacks.
Specifically, in our experiments, we set $\log q=478, \log p=460, \log n=15$.
We choose $\chi_s$ as ternary distribution (i.e., each coefficient is sampled uniformly from
$\{-1,0,1\}$) with parameter 64 and choose $\chi_e$ as Gaussian distribution with mean 0 and variance 1.22.
According to the standard of HE in \cite{Albrecht2021}, the security level of our model is $\lambda=256$.

\subsection{Effectiveness of \nameScheme's Packing}\label{ssec:impactPack}

 \begin{figure*}[t]
    \hspace{-0.45cm}
	\subfloat[FCN client side]
	{
		\begin{minipage}[b]{.33\textwidth}
			\centering
			\includegraphics[width=2.45in]{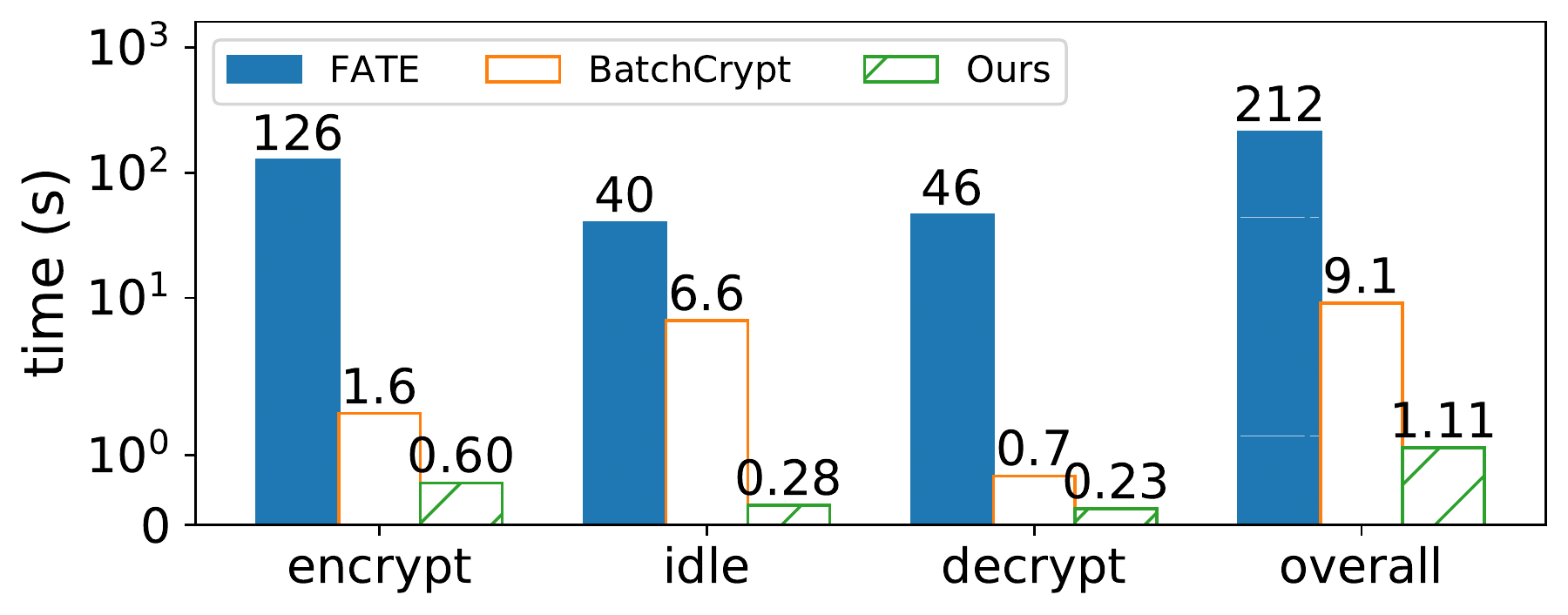}
		\end{minipage}
	}
	\subfloat[AlexNet client side]
	{
	\begin{minipage}[b]{.33\textwidth}
		\centering
		\includegraphics[width=2.4in]{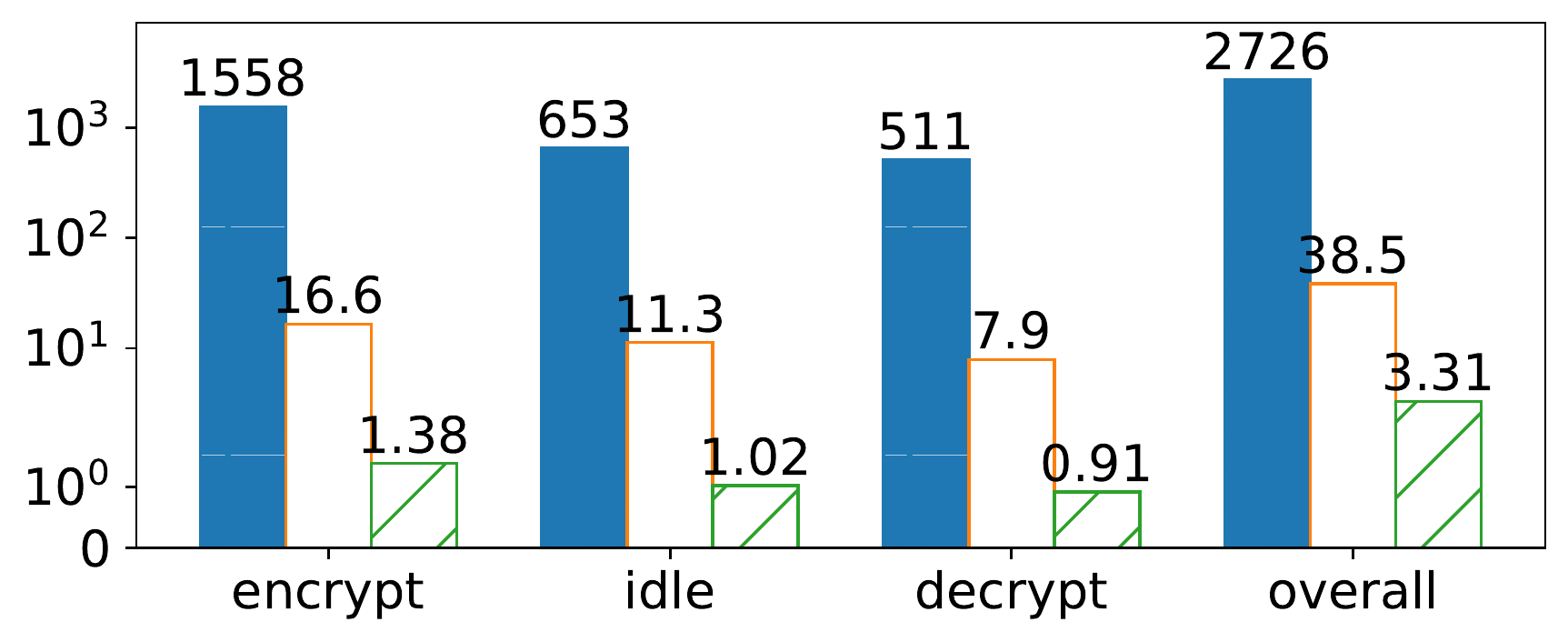}
	\end{minipage}
	}
	\subfloat[LSTM client side]
	{
	\begin{minipage}[b]{.33\textwidth}
		\centering
		\includegraphics[width=2.4in]{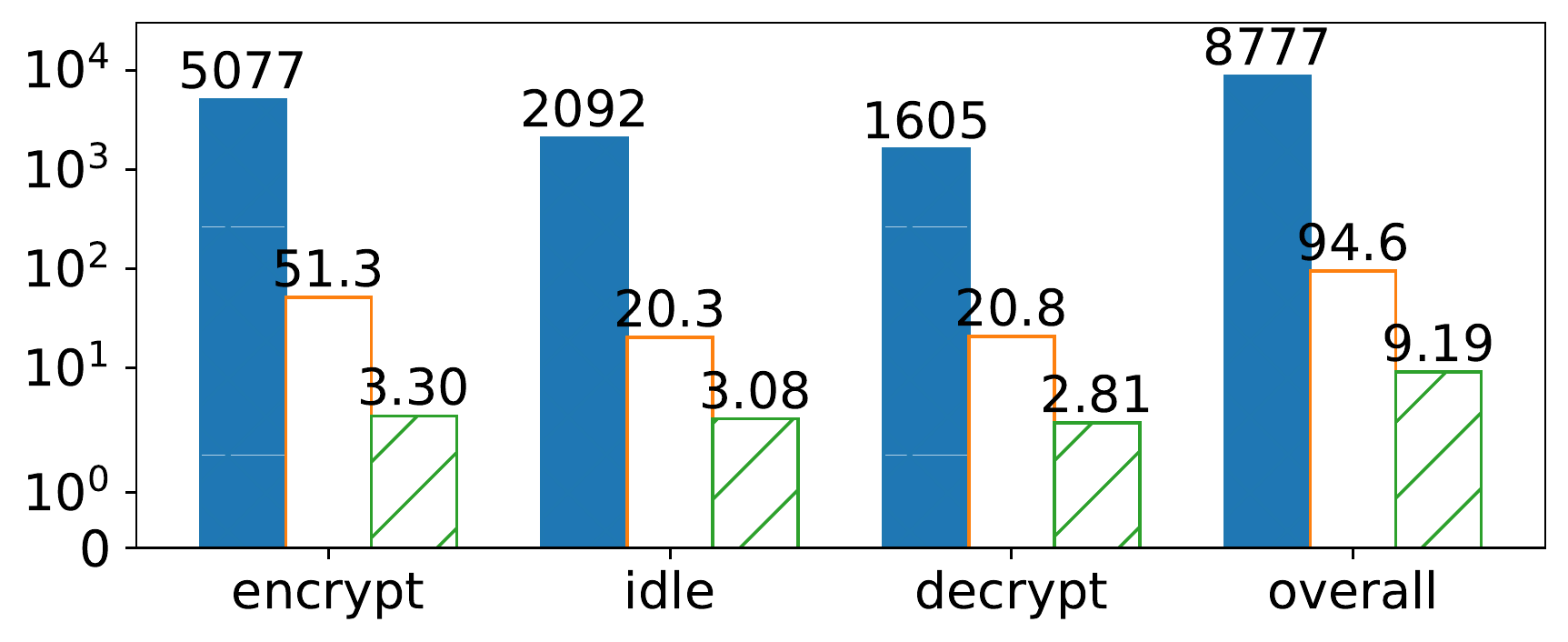}
	\end{minipage}
	}
  \\
	\subfloat[FCN server side]
	{
		    \hspace{-0.45cm}
		\begin{minipage}[b]{.33\textwidth}
			\centering
			\includegraphics[width=2.45in]{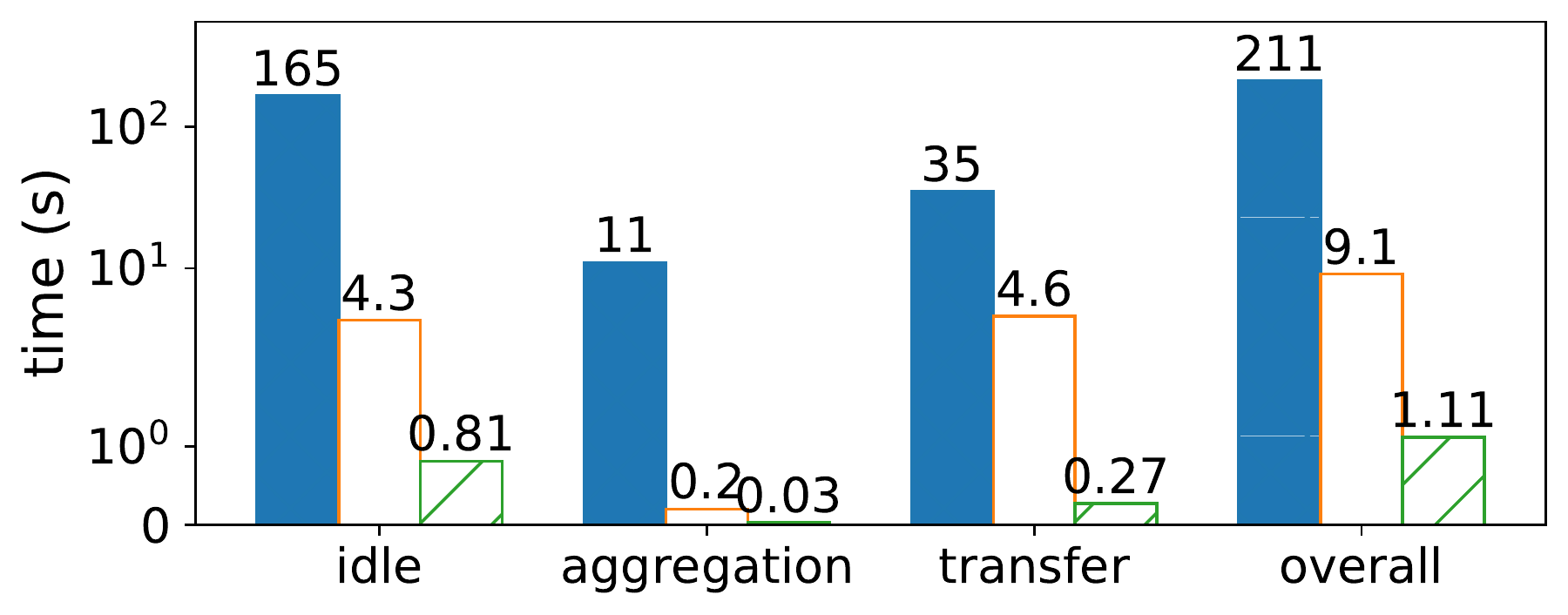}
		\end{minipage}
	}
	\subfloat[AlexNet server side]
	{
		\begin{minipage}[b]{.33\textwidth}
			\centering
			\includegraphics[width=2.4in]{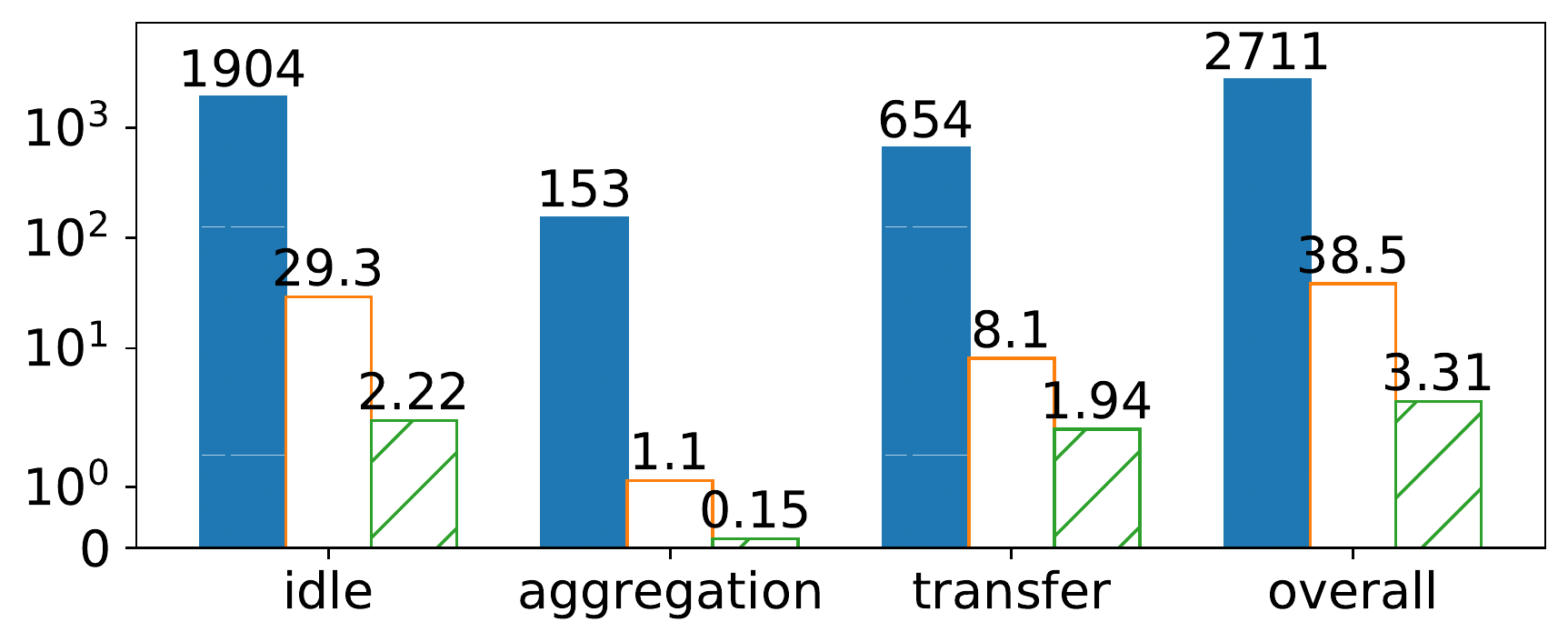}
		\end{minipage}
	}
	\subfloat[LSTM server side]
	{
		\begin{minipage}[b]{.33\textwidth}
			\centering
			\includegraphics[width=2.4in]{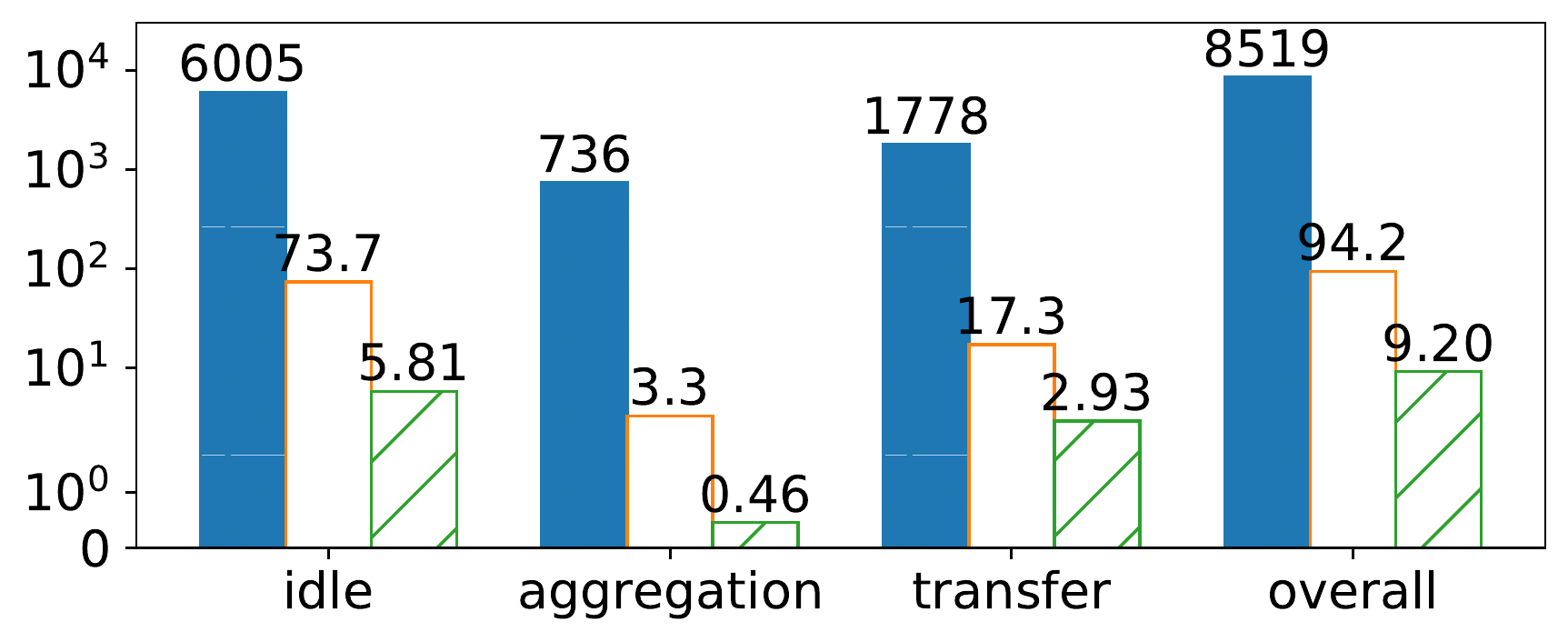}
		\end{minipage}
	}
	\caption{Comparisons of training iteration time among FATE, BatchCrypt, and \nameScheme, where ``idle" represents the idle waiting time of participates (the clients or the server), and ``encrypt" in \nameScheme measures the time of the pseudo-random generation, encoding, packing, and encryption on the client. All models are trained for 50 iterations and the run time is the average result.}
	\label{fig:runtime}
\end{figure*}

The effectiveness of \nameScheme's packing lies in two folds.

(1) It reduces the computation cost and communication traffic caused by encryption.
To compare computation cost, we train \nameScheme with and without packing and train the plain baseline until they converge. %
Fig.~\ref{fig:unpkVSpack} (a), (b), and (c) depict the training speed of the three models, and Fig.~\ref{fig:unpkVSpack} (d), (e), and (f) present the model quality of the three trained models at different iterations.
These figures show that \nameScheme's packing mitigates the training speed delay caused by encryption (\nameScheme's encryption without packing) by more than $3\times$, and \nameScheme with packing only slows down the training time by less than $2\times$ for the FCN, AlexNet, and LSTM models compared with their plain baseline models which require no encryption and thus are the fastest to converge.
Table~\ref{tab:compPlainCipher} presents the number of messages and the communication traffic of the plain FL (plain), \nameScheme without packing (unpk), and \nameScheme with packing (pack).
It can be seen from the table that \nameScheme's packing method can pack about 20 messages into one ciphertext, and \nameScheme with packing significantly decreases the number of ciphertexts and the communication traffic caused by encryption (cipher traffic without packing).

\begin{table}
	\centering
	\caption{Comparisons of communication traffic for one iteration on entities between \nameScheme with and without packing.}
	\label{tab:compPlainCipher}
	\begin{tabular}{|l|l|l|l|}
		\hline
		& FCN     & AlexNet   & LSTM     \\\hline
		\#Gradients (plain) & 101.77K & 1.25M     & 4.02M    \\\hline
		\#Ciphertexts (unpk) & 7 & 77 & 246    \\\hline
		\#Ciphertexts (pack) & 1 & 4  & 12    \\\hline
        Plain traffic (plain) & 0.78MB  & 9.54MB  & 30.67MB \\\hline
		Cipher traffic (unpk) & 26.15MB & 287.60MB & 918.81MB \\\hline
		Cipher traffic (pack) & 1.87MB  & 14.94MB  & 44.82MB \\\hline		
	\end{tabular}
\end{table}

(2) It reduces the impact of \nameHE's encoding on model accuracy/loss.
Table~\ref{tab:decAcc} presents the magnitude of error of the three benchmark models after encryption with and without polynomial packing. 
An error of a model is calculated by averaging the differences between the decrypted aggregation gradients and the aggregated original gradients.
We can see from the table that when packing is used, the errors of the three models are all close to 0. On the contrary, if packing is not used, their errors are quite large.

The model quality of the plain baseline and \nameScheme with packing are nearly the same.
Fig.~\ref{fig:unpkVSpack} depicts the experimental results of test accuracy with and without packing.
For FCN model trained with FMNIST dataset, plain baseline reaches peak accuracy $89.15\%$ at the 36th epoch, while \nameScheme without and with packing reach $89.01\%$ and $89.05\%$ at the 36th epoch and 34th epoch, respectively.
For AlexNet model trained CIFAR10 dataset, plain baseline reaches peak accuracy $73.35\%$ at the 265th epoch, while \nameScheme  without and with packing reach $70.14\%$ and $73.85\%$ at the 265th epoch and the 255th epoch, respectively.
Finally, for LSTM, plain baseline reaches bottom loss $0.0336$ at the 20th epoch, while \nameScheme without and with packing reach $0.0358$ and $0.0337$ at the 21st epoch and the 20th epoch, respectively.
Therefore, with the packing technique, the encoding of \nameHE brings negligible negative impact on the final model quality of \nameScheme.

\begin{table}[!h]
	\centering
	\caption{The magnitude of difference between the decrypted messages and the original messages for FCN, AlexNet, and LSTM models a quantification bit width of 32.	}
	\label{tab:decAcc}
	\begin{tabular}{|l|l|l|l|}
		\hline
		     &FCN        & AlexNet   & LSTM      \\ \hline
		unpk & $10^{-2}$ & $10^{-2}$ & $10^{-2}$ \\ \hline
		pack & $10^{-9}$ & $10^{-9}$ & $10^{-9}$ \\ \hline
	\end{tabular}
\end{table}

\subsection{Effectiveness of \nameScheme}\label{ssec:efficiency}

\textbf{\nameScheme compared with BatchCrypt and FATE.}
We next compare the effectiveness of \nameScheme with BatchCrypt and FATE in terms of running time and network traffic.
The quantization bit width is set to 16, which is the same as BatchCrypt and FATE in \cite{Zhang2020BatchCrypt}.
We run the experiments of \nameScheme for 50 iterations and show their averaged result against the same measurements of BatchCrypt and FATE in Figs.~\ref{fig:runtime} and \ref{fig:traffic}.
From Fig.~\ref{fig:runtime}, we can see that compared with BatchCrypt and FATE, \nameScheme improves the computational efficiency of encryption by about one order of magnitude and two orders of magnitude, respectively.
The ``idle" time for transmitting ciphertexts in \nameScheme is also reduced because of the reduction of the network traffic.
Fig.~\ref{fig:traffic} compares the specific network traffic of one training iteration among \nameScheme, BatchCrypt, and FATE.
We see that \nameScheme reduces the communication traffic of FATE by up to $68\times, 103\times$, and $112\times$ for FCN, AlexNet, and LSTM, respectively, and reduces that of BatchCrypt by up to 1MB, 60MB, and 50MB for FCN, AlexNet, and LSTM, respectively.

\begin{figure}[htbp]
\hspace{-0.45cm}
\subfloat[FCN]{
 \begin{minipage}[b]{.17\textwidth}
  \centering
  \includegraphics[width=1.1in]{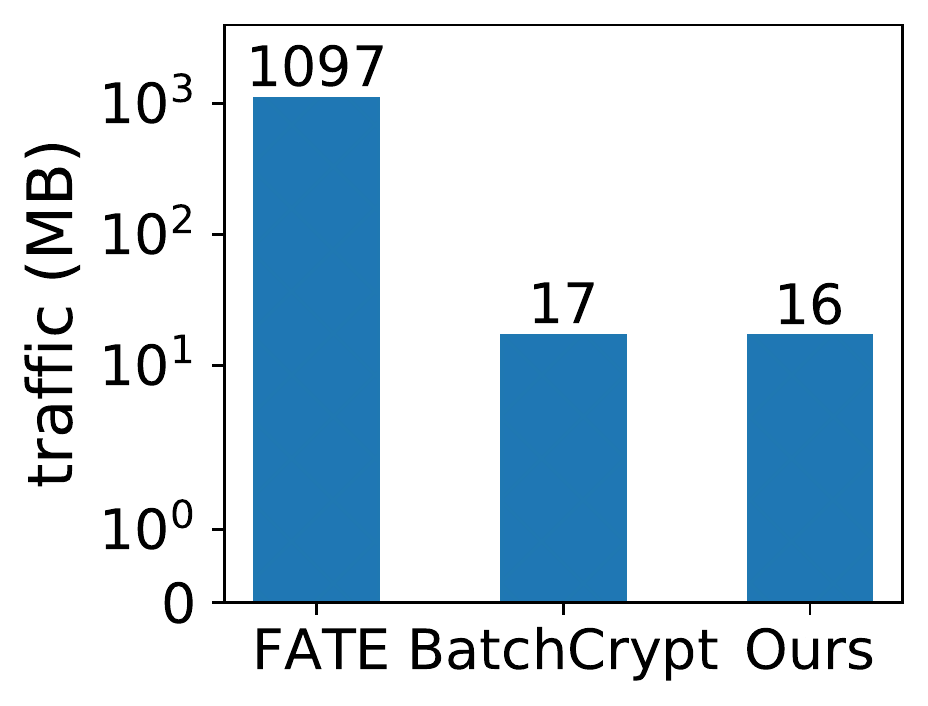}
 \end{minipage}
 }
\subfloat[AlexNet]{
  \hspace{-0.45cm}
 \begin{minipage}[b]{.17\textwidth}
  \centering
  \includegraphics[width=1.1in]{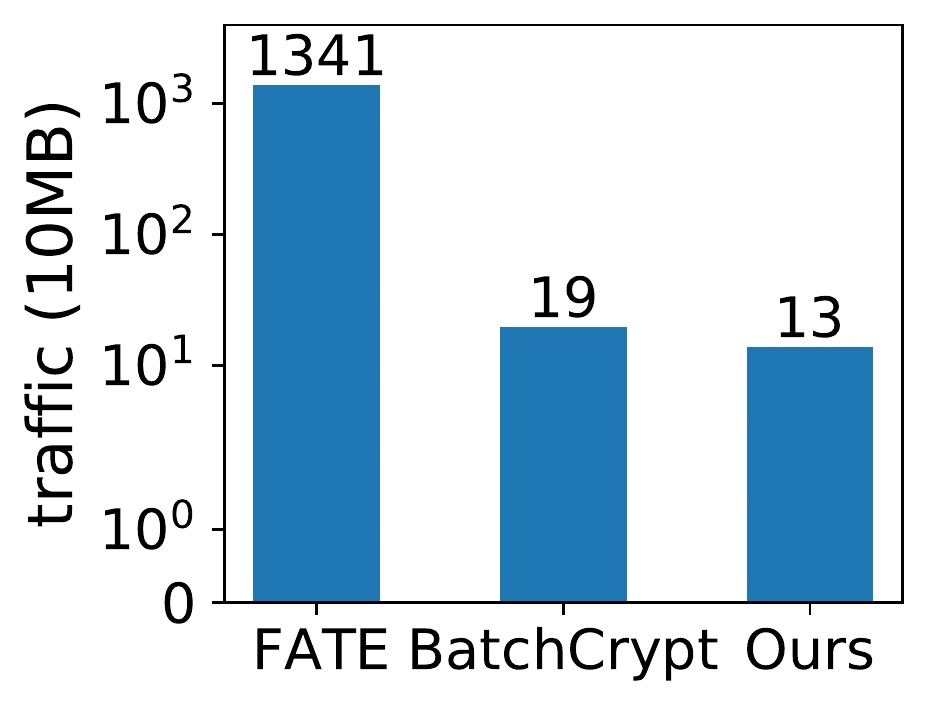}
 \end{minipage}
 }
\subfloat[LSTM]{
  \hspace{-0.45cm}
 \begin{minipage}[b]{.17\textwidth}
  \centering
  \includegraphics[width=1.1in]{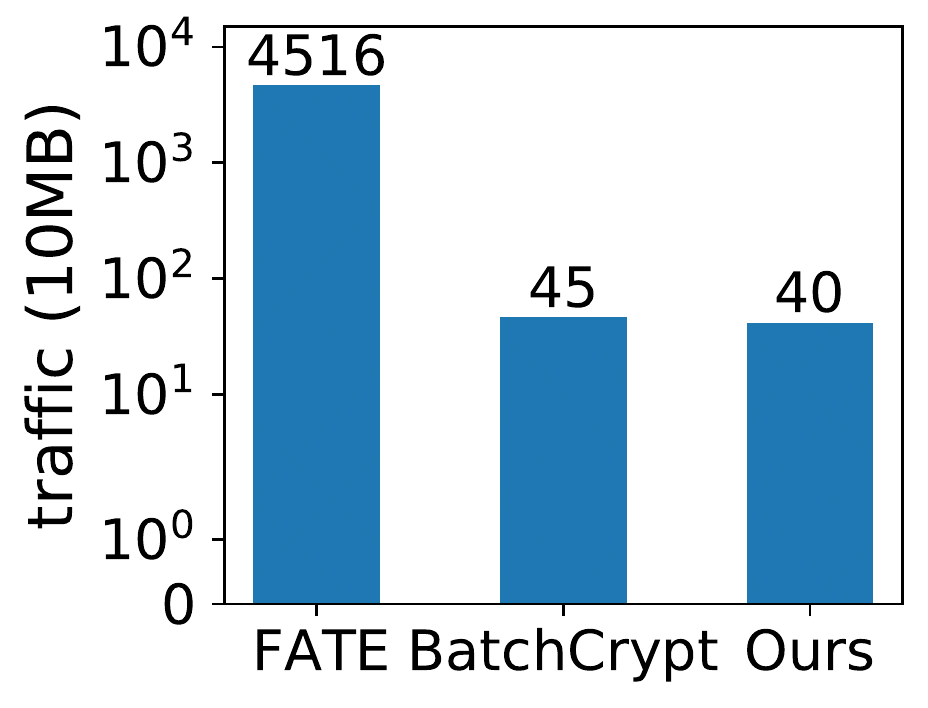}
 \end{minipage}
 }
\caption{Network traffic comparisons among FATE, BatchCrypt, and \nameScheme in one training iteration.}
 \label{fig:traffic}
\end{figure}

\textbf{\nameScheme versus Plaintext Learning.}
Next, we compare \nameScheme with the plain FL that involves no encryption and thus is an ideal baseline offering the optimal performance.
We present the training time and communication traffic for one iteration in Fig.~\ref{fig:totalTimeTraffic}.
From the figure we know that while encryption remains the major performance bottleneck, \nameScheme successfully reduces both the computation and communication burdens, making it practical to obtain the same trained model as the plain FL.

\begin{figure}[htbp]
\hspace{-0.55cm}
\subfloat[Time]{
	\begin{minipage}[b]{.25\textwidth}
		\centering
		\includegraphics[width=1.75in]{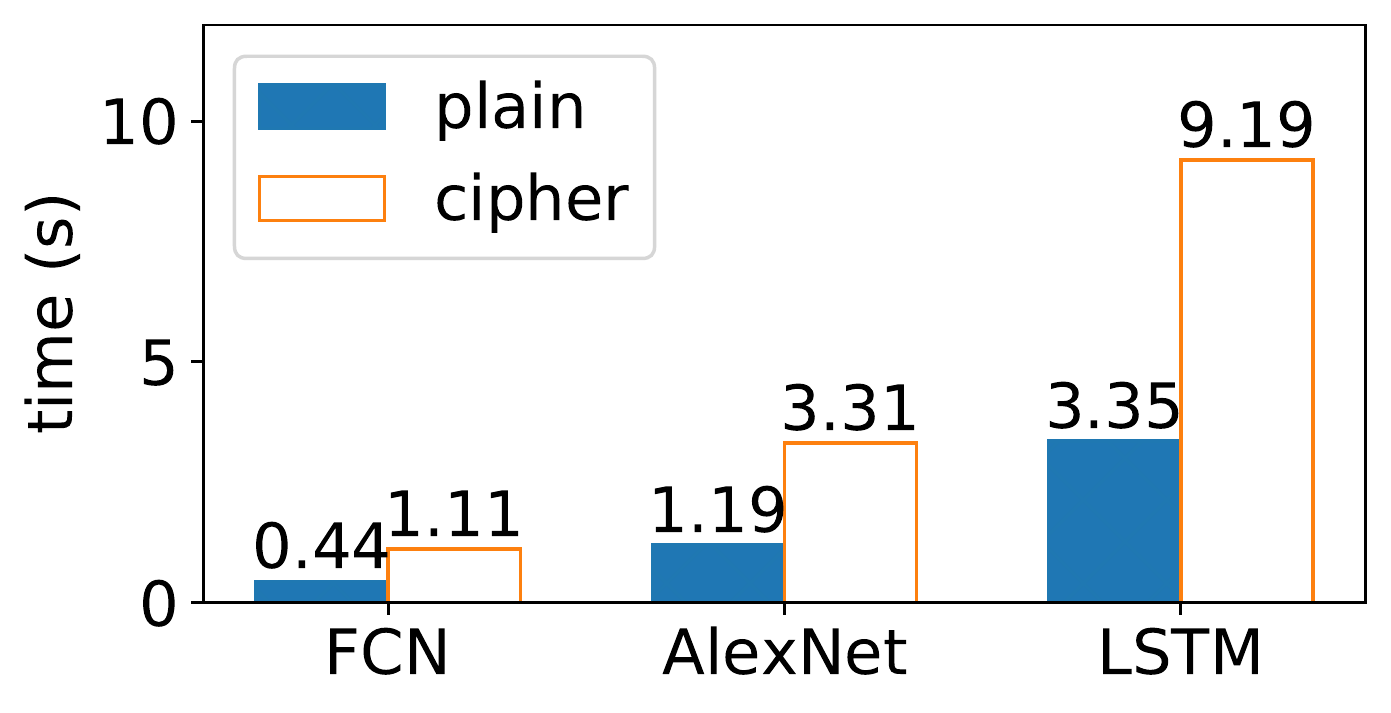}
	\end{minipage}
}
\subfloat[Communication]{
\begin{minipage}[b]{.25\textwidth}
	\centering
	\includegraphics[width=1.85in]{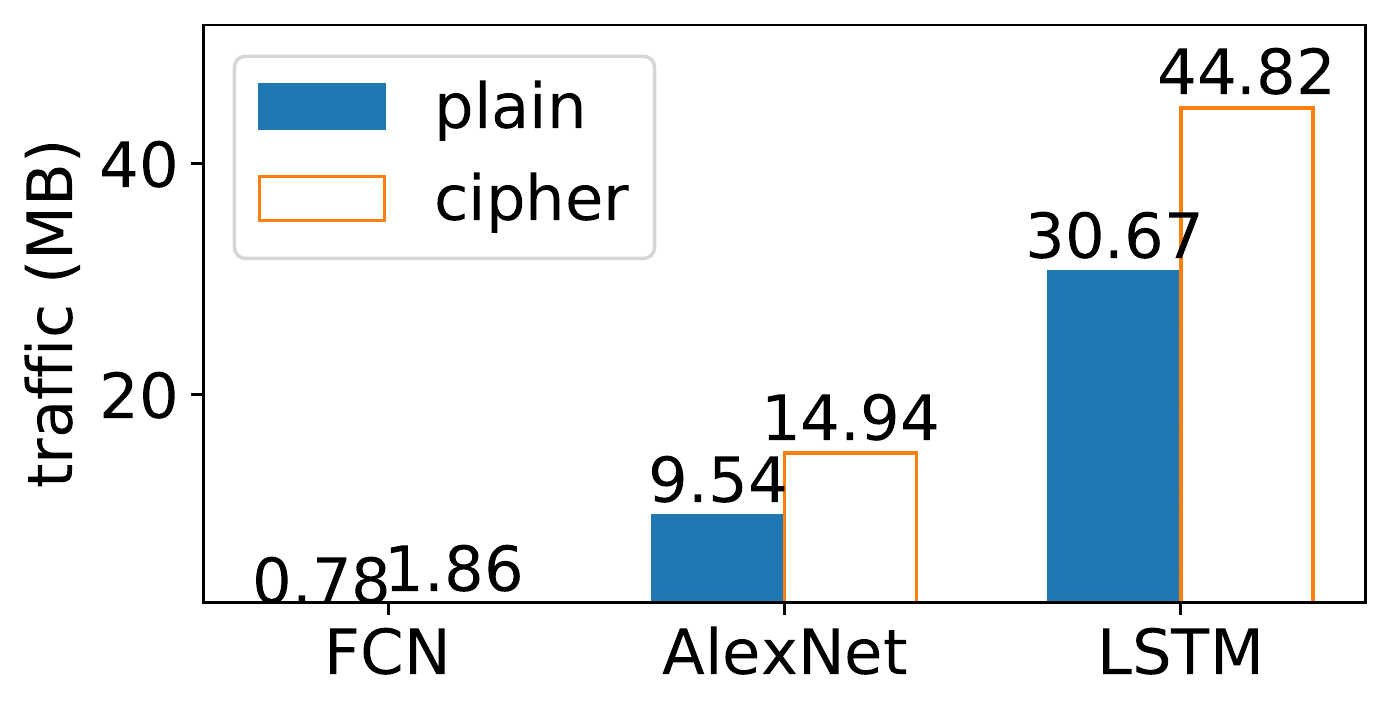}
\end{minipage}
}
 \caption{Comparisons of time and communication costs on clients between \nameScheme and plain FL for one iteration. }
  \label{fig:totalTimeTraffic}
\end{figure}

\begin{table}[ht]
\scriptsize
\centering
\caption{Comparisons of total training time and network traffic for converged FCN, AlexNet, and LSTM models under different training modes.}
\label{tab:comp_effcy}
\begin{tabular}{|l|l|l|l|l|l|}
\hline
Model                    & Mode       & Epochs & Acc./loss & Time(h) & Traffic(GB) \\ \hline
\multirow{4}{*}{FCN}     & FATE       & 40     & 88.62\%   & 122.5   & 2228.3      \\ \cline{2-6}
                         & BathCrypt  & 68     & 88.37\%   & 8.9     & 58.7        \\ \cline{2-6}
                         & \nameScheme& 34     & 88.79\%   & 0.6    & 28.9        \\ \cline{2-6}
                         & Plain      & 36     & 89.15\%   & 0.2    & 12.8       \\ \hline
\multirow{4}{*}{AlexNet} & FATE       & 285    & 73.97\%   & 9495.6  & 16422.0     \\ \cline{2-6}
                         & BathCrypt  & 279    & 74.04\%   & 131.3   & 277.8       \\ \cline{2-6}
                        & \nameScheme & 255    & 73.85\%   & 10.2   & 180.6       \\ \cline{2-6}
                        & Plain       & 265    & 73.35\%   & 3.8     & 119.8       \\ \hline
\multirow{4}{*}{LSTM}    & FATE       & 20     & 0.0357    & 8484.4  & 15347.3     \\ \cline{2-6}
                        & BathCrypt   & 23     & 0.0335    & 105.2   & 175.9       \\ \cline{2-6}
                        & \nameScheme & 20     & 0.0351    & 8.9     & 139.9       \\ \cline{2-6}
                        & Plain       & 20     & 0.0357    & 3.3     & 100.4        \\ \hline
\end{tabular}
\end{table}

\textbf{Training to Convergence.}
We next compare the total training time and network traffic of FATE, BatchCrypt, and \nameScheme until they are trained to convergence.
Because the whole training process requires exceedingly high costs in both time and traffic, we simulate distributed training of FL locally until the model converges, and estimate the total time and communication cost needed for convergence according to the network bandwidth in Table~\ref{tab:bandwidth} and the number of iterations in the local training.
Table~\ref{tab:comp_effcy} presents the estimation results of FATE, BatchCrypt, \nameScheme, and plain FL.
Compared with FATE and BatchCrypt, \nameScheme lowers both training time and network traffic required for convergence.
Specifically, \nameScheme reduces the training time of FATE by $204\times,930\times$, and $953\times$ for FCN, AlexNet, and LSTM, respectively, and reduces that of BatchCrypt by $14\times, 12\times$, and $11\times$, respectively.
In the meantime, \nameScheme shrinks the network traffic of FATE by $77\times, 90\times,109\times$ and $2\times,1.5\times,1.25\times$, respectively.
Besides, compared with plain FL, \nameScheme slows down the training time towards convergence by $2\times, 1.68\times$, and $1.69\times$ for FCN, AlexNet, and LSTM, respectively.
It expands the communication traffic of plain FL by less than $1.5\times (1.26\times,0.5\times, 0.39\times)$ for the three models.
Therefore, \nameScheme significantly reduces both the computation and communication burden caused by encryption, enabling efficient secure multiparty computation for cross-silo FL.


\section{Conclusion and Future Work}\label{sec:conclusion}

In this paper, we propose \nameScheme, an efficient secure cross-silo federated learning model.
To protect the confidentiality of the clients' datasets and ensure the security of the training results under honest-but-curious entities, we develop an efficient and secure additively homomorphic encryption algorithm \nameHE based on RLWE hard problem.
\nameHE can provide privacy protection without the need for secure channel transmission of ciphertext.
To improve the computational efficiency and reduce the communication cost of the existing MKHE algorithms, we endow \nameHE with a one-step decryption process and a polynomial packing method. 
Besides, FFT is also used to accelerate polynomial multiplications.
Our theoretical analyses and experimental evaluations present that \nameScheme model can achieve security and efficiency without reducing the model accuracy.

In our future work, we plan to extend our \nameHE to fully homomorphic encryption and plug it into cross-silo FL to achieve secure nonlinear aggregation with which FL can be more robust to heterogeneous clients with different dataset quantities, dataset distribution, computing power, communication bandwidth, etc.

\bibliographystyle{IEEEtran}
\bibliography{IEEEabrv,cited}


\end{document}